\documentclass[12pt]{article}
\pdfoutput=1

\usepackage{amssymb,amsmath,mathtools,geometry,float,graphicx,color,setspace,booktabs,algorithm,algpseudocode}

\usepackage[square,longnamesfirst]{natbib}

\usepackage{titlesec}

\definecolor{mycolor}{RGB}{128, 0, 32}
\usepackage[colorlinks=true,linkcolor=mycolor,citecolor=mycolor,urlcolor=mycolor]{hyperref}

\newtheorem{theorem}{Theorem}

\newtheorem{lemma}{Lemma}
\newtheorem{definition}{Definition}
\newenvironment{proof}[1][Proof]{\noindent\textbf{#1.} }{\ \rule{0.5em}{0.5em}}

\begin{document}

\begin{titlepage}
\title{Approximate Minimax Estimation of a Bounded Normal Mean via Stochastic Mirror Ascent\thanks{We would like to thank Tim Armstrong and Paul Delatte for a very helpful discussion that sparked our interest in this problem. We would also like to thank Jeffrey Negrea and Alexander Terenin for useful exchanges regarding different  computational aspects of the Bounded Normal Mean problem and several suggestions for future work. We received very helpful feedback from Andr{\'e}s Aradillas-Fernandez, Jos\'e Blanchet, Iain Johnstone, Chen Qiu, Francesca Molinari, Mikkel Plagborg-M{\o}ller, J{\"o}rg Stoye, Lezhi Tan, and Christian Wolf. Montiel Olea gratefully acknowledges financial support by the National
Science Foundation Grant SES-2315600. }}
\author{Jos\'e Luis Montiel Olea\thanks{Department of Economics, Cornell University. Email: {{\href{mailto:montiel.olea@gmail.com}
{montiel.olea@gmail.com}}}} \and Ekaterina Zubova\thanks{Department of Economics, Cornell University. Email: {{\href{mailto:ez268@cornell.edu}
{ez268@cornell.edu}}}}}
\date{\today}
\maketitle
\begin{abstract}
\noindent This paper presents a computational approach to find an approximately minimax estimator for the classical \emph{Bounded Normal Mean} problem. The suggested procedure is the Bayes estimator corresponding to an approximately least-favorable distribution  obtained from a \emph{stochastic mirror ascent} routine for concave maximization. The paper shows that both the approximately least-favorable distribution and the approximately minimax estimator are indeed close (in a sense we make precise) to their desired targets. Simulation evidence suggests that the approximately minimax estimator can yield---with a reasonable amount of compute---risk improvements from 6\% to almost 18\% relative to the minimax linear estimator (which is known to admit a maximal improvement of 20\%). The approximately minimax estimator is then applied to the problem of how to best aggregate the information contained in local projections and vector autoregressions to estimate an impulse response coefficient.    \\
\vspace{0in}
\vspace{0in}

\bigskip
\end{abstract}
\setcounter{page}{0}
\thispagestyle{empty}
\end{titlepage}
\pagebreak \newpage

\doublespacing

\section{Introduction} \label{sec:introduction}

We propose a numerical algorithm to find an approximately minimax estimator (under quadratic loss) of the mean of a standard normal distribution, when such mean is known to belong to the interval $[-m,m]$. \cite{casella_strawderman_1981} refer to this classical problem in statistical decision theory as the \emph{Bounded Normal Mean} estimation problem. As noted by \cite{johnstone1992minimax}, the minimax risk in this problem plays a central role in comparing the behavior of linear to nonlinear estimators in an important class of nonparametric estimation problems; see \citet{ibragimov1985nonparametric} and  \citet{donoho_liu_macgibbon_1990}. The minimax risk of the Bounded Normal Mean problem has also been used recently by \cite{armstrong_kline_sun_2025} as an input to construct estimators that \emph{adapt to misspecification}.

Our suggested algorithm finds an approximately minimax estimator by first obtaining an approximately \emph{least-favorable distribution}. This strategy is common in the literature; see  \cite{kempthorne1987numerical} and \cite{chamberlain2000econometric}. One point of departure relative to previous work is that we use a \emph{stochastic mirror ascent routine} (with negative entropy as a mirror map) to approximately solve the following concave optimization problem: we maximize Bayes risk over the space of probability distributions supported on a fixed (equally-spaced) grid of $I$ points over the interval $[-m,m]$. In what follows, we let $\Delta^{I-1}$ denote the space of such distributions, without making explicit reference to the support points. 

Before providing more specific details about our contribution, it is worth noting that our numerical strategy has two sources of approximation error, both of which are explicitly taken into account in our analysis. The first source of error is an obvious \emph{discretization error}, which comes from the fact that we optimize Bayes risk over the simplex $\Delta^{I-1}$ and not over \emph{all} probability distributions supported on $[-m,m]$. We show that we can control this discretization error thanks to a special (Lipschitz) continuity property of the risk function (see Lemma \ref{lemma:RiskFunction is Lips} and Lemma \ref{lemma:discretized vs original}). The second source of error is what we term  \emph{optimization error}, which comes from the fact that the mirror ascent routine for concave maximization is designed to find an \emph{approximate solution}, and not an \emph{exact} one. An explicit analysis of this type of error is common in theoretical computer science, operations research, and optimization where it is typical \emph{``to relax the requirement of finding an optimal solution, and instead settle for a solution that is good enough''}  \cite[p. 14]{shmoys2011design}. We can control this optimization error thanks to the rich literature that has analyzed the theoretical properties of mirror descent (or ascent). 

The methods of mirror descent \citep[Chapter 3]{nemirovski_yudin_1983} are a family of iterative procedures recommended in the optimization literature for approximately solving convex problems of high dimensions. These methods are iterative, first-order optimization algorithms, in that they typically require repeated evaluations of the objective function and an estimator of its (sub or super) gradient, but do not exploit any further smoothness information about the objective function.\footnote{The algorithm we recommend only requires repeated evaluations of an estimator of the supergradient of the objective function, but not of the objective function itself. See Algorithm \ref{alg:SMA}.} The optimization error associated with these algorithms comes from the fact that the procedure is designed to stop after finitely many iterations. Applications of the methods of mirror descent to other problems in statistics are given in \cite{juditsky2008} and \cite{fernandez2025approximate}.

There are at least four reasons that explain why mirror ascent was the algorithm we chose for finding an approximately least-favorable distribution. \emph{First}, even in the simple Bounded Normal Mean problem, it is difficult to get an explicit characterization of higher-order derivatives of the objective function defining the least-favorable distribution, which means that first-order methods (which only use information concerning the gradient and do not require further smoothness assumptions) are particularly attractive. \emph{Second}, there is also a sense in which our algorithm is essentially the best first-order iterative algorithm for approximately minimizing (maximizing) a convex (concave), Lipschitz function over the simplex (as in the problem defining the approximately least-favorable distribution); see Proposition 4.2 in \cite{ben2001ordered}.\footnote{It is also known that the rate of convergence of mirror descent (or ascent) improves over regular subgradient descent for convex minimization (maximization) problems in the probability simplex \citep[Section 4.3]{bubeck_2015}.} \emph{Third}, as our results show, the supergradient of the objective function at any distribution $\pi \in \Delta^{I-1}$ depends on the risk function $R(d_{\pi},\theta)$, where $d_\pi$ is the Bayes estimator associated with $\pi$. This means that even though evaluating the supergradient is potentially costly (as it requires numerical integration), we can rely on the well-known stochastic version of mirror ascent that replaces the supergradient by an unbiased estimator (and in our case, this unbiased estimator can be constructed using only \emph{one} Monte Carlo draw per grid point per iteration, further reducing the computational burden of the algorithm). \emph{Fourth}, since there is a rich theory regarding the convergence analysis of stochastic mirror descent (ascent) \citep[Chapter 6]{bubeck_2015}, we can provide explicit theoretical guarantees about the performance of our suggested algorithm.\footnote{While it is entirely possible that off-the-shelf methods implemented in numeric computing platforms (for example,  second-order optimization methods like sequential quadratic programming via \texttt{fmincon} in Matlab) perform as satisfactorily as our recommended algorithm, it is not clear to us how to provide theoretical guarantees without being more explicit about the tuning parameters used by the algorithm to generate approximate minimax solutions.} 

The first theoretical result in this paper (Theorem \ref{thm:avg_prior_pointwise_LF}) shows that our stochastic mirror ascent algorithm \emph{provably} finds an approximately least-favorable distribution for the Bounded Normal Mean problem after \emph{finitely} many iterations. The theoretical guarantees we provide are expressed in terms of a scalar parameter $\epsilon>0$. This parameter controls the overall approximation error of our numerical algorithm. Let $\lceil \cdot \rceil$ denote the ceiling function.\footnote{The ceiling function returns the smallest integer that is greater than or equal to a given number.} The theorem states that if the number of iterations, $J(\epsilon)$, and the number of points in the grid, $I(\epsilon)$, are chosen to be
\[ J(\epsilon) := \left \lceil \frac{2 M^2 \ln (I(\epsilon))}{(\epsilon/2)^2} \right \rceil, \quad I(\epsilon):= \left \lceil 1 + \frac{3m^3 + 4m^2}{\epsilon/2} \right \rceil, \quad \textrm{ where }M:= 4m^2, \]
then, for any $\alpha \in (0,1)$, with probability at least $1-\alpha$, the Bayes risk associated with the prior distribution $\pi^\epsilon$ found by the algorithm is close to the \emph{maximin} value of the Bounded Normal Mean problem. And, in particular, it differs at most by an additive factor  
\[ \epsilon \left(1 + \sqrt{\frac{\ln(1/\alpha)} {\ln(I(\epsilon))}} \right).\]
With the exception of $\epsilon$, our theorem makes explicit recommendations about all of the parameters needed to run the algorithm (such as the initial condition and the learning rate). It is worthwhile to make three remarks about our result. \emph{First}, the probabilistic statement that appears in the theorem comes from the fact that our theoretical analysis acknowledges the fact that the gradient of the objective function needs to be estimated by Monte Carlo integration. This estimation induces---just as any other stochastic approximation method---an additional source of variation that will make our approximately least-favorable distribution change for different \emph{runs} of the algorithm if the random seed is not fixed. \emph{Second}, the total number of iterations (which scales logarithmically in the number of grid points) should be viewed as \emph{sufficient} to guarantee a good performance of the algorithm, but it is by no means \emph{necessary}. In fact, our simulation results show that the approximation error is likely to be of size $\epsilon$ for a considerably smaller number of iterations (which is important, given that $J(\epsilon)$ is already in the order of millions for values of $m$ as small as $m=3$). \emph{Third}, the convergence guarantee we provide in Theorem \ref{thm:avg_prior_pointwise_LF} is expressed in terms of how well we can approximate the \emph{maximin} value of the Bounded Normal Mean problem, and not in terms of how close the approximate least-favorable distribution is relative to the true least-favorable distribution.\footnote{The least-favorable distribution is known to exist and be unique; see Proposition 4.19 in \citet{johnstone_2019}. See also \citet{delatte26}. }  

The second theoretical result of this paper (Theorem \ref{thm:convergence of pi_n}) complements Theorem \ref{thm:avg_prior_pointwise_LF} by showing that $\pi^{\epsilon}$ will indeed be close (in terms of the \emph{1-Wasserstein distance}) to the least-favorable distribution of the Bounded Normal Mean problem, provided $\epsilon$ is small enough.  To the best of our knowledge, none of the other existing papers in the literature that have suggested algorithms for the Bounded Normal Mean problem---such as \citet{feldman1989minimax}, \citet{donoho_liu_macgibbon_1990}, and \citet{gourdin1994global}---have derived results like ours. This is also true for other existing papers that have suggested algorithms for finding approximate solutions for more general minimax/maximin problems---such as \citet{kempthorne1987numerical}, \citet{johnstone1992minimax}, \citet{chamberlain2000econometric}, \citet{montielolea_aradillasfernandez_blanchet_qiu_stoye_tan_2024}, \citet{fernandez2025approximate}, \citet{elliott2015nearly}, and \citet{guggenberger2025numerical}.

The third theoretical result of this paper (Theorem \ref{thm:minimax sequence}) shows that there is a sense in which the Bayes estimator associated with the prior distribution $\pi^\epsilon$ found by the stochastic mirror ascent algorithm is approximately minimax. The notion of approximate minimaxity that we use in this paper extends the classical definition of \emph{minimax sequence} given in the seminal work of \cite{ghosh1964uniform}, p. 1037. Broadly speaking, \cite{ghosh1964uniform} defines a deterministic sequence of decision rules to be a \emph{minimax sequence} if their worst-case risk converges to the minimax risk of the problem of interest. Since, as we have explained before, our stochastic mirror ascent algorithm introduces noise beyond the statistical uncertainty contained in the estimation problem, we need to allow for the possibility of a stochastic sequence of decision rules that approximate the minimax value of interest. In light of this observation, we define a sequence of decision rules to be a \emph{stochastic minimax sequence} if their worst-case risk converges in probability to the minimax value of interest. This means that the \emph{tail} of a stochastic minimax sequence contains decision rules that are $\epsilon$-minimax---in the sense of \cite{ferguson_1967}, Chapter 1, Definition 4, p. 33---with probability at least $1-\delta$, for any $\epsilon,\delta>0$.  

The last theoretical result of the paper (Theorem \ref{thm:rule_convergence}) shows that the Bayes estimator associated with the prior distribution $\pi^\epsilon$ approximates the \emph{exact} minimax estimator of a bounded normal mean in a stronger sense: for any compact set $K$ of data realizations, there is a small enough $\epsilon$ (that could depend on the set) such that the  approximate and the exact minimax estimators are \emph{uniformly} close to each other over $K$, with probability approaching one. This means that our approximate minimax estimator converges---\emph{uniformly over compact sets}---to the minimax estimator of the Bounded Normal Mean problem.  

{\scshape Simulations:} It is known that in the Bounded Normal Mean problem the worst-case risk of the best minimax \emph{linear} estimator is at most 25\% above the best, unconstrained, minimax estimator. This is the so-called \cite{ibragimov1985nonparametric} constant, see Theorem 4.7 in \cite{johnstone_2019}. The minimax linear estimator thus provides a computationally feasible (and trivial to compute) benchmark for the Bounded Normal Mean problem. The value of the \cite{ibragimov1985nonparametric} constant implies that the largest reduction in risk that we could expect to see from our approximately minimax estimator (relative to the minimax linear estimator) is 20\%. Figures \ref{fig: risk_all_m_4} and \ref{fig: improvement_all} show that our stochastic mirror ascent routine generates risk improvements over the minimax linear estimators that range from about 6\% to 18\%, as $m$ ranges from 1 to 4 (with the largest improvement occurring at $m=1.6$). In our simulations we choose $\epsilon$ to be $\epsilon(m) := (1/5)m^2/(1+m^2)$. We provide a rationale for this choice in Section \ref{sec: simulation}.   

In terms of the computational effort required to achieve the 6\% to 18\% improvement reported above, we note that the largest reduction in risk that we see takes about 45 seconds on a personal computer. The computational time increases with $m$, and implementing our algorithm for $m=3$ takes about 2 hours (see Figure \ref{fig: time_SMD_new}). The main bottleneck we face in the implementation of our algorithm is the number of iterations recommended in Theorem \ref{thm:avg_prior_pointwise_LF}. We report simulations where we bound the \emph{wall-clock} time of our procedure to 5 and 15 minutes. With the 15-minute limit, we still see improvements  between 6\% and 18\%, but the improvement over linear estimators when $m=3$ drops from about 14\% without constraints, to almost 6\%. 

{\scshape Application:} A common problem in macroeconometrics is the estimation of the dynamic causal effects---or \emph{impulse response functions}---of aggregate shocks to macroeconomic policies or economic fundamentals. The two most common empirical methods currently used by applied macroeconomists for doing so are \emph{structural vector autoregressions} (henceforth interchangeably referred to as VARs or SVARs), going back to the seminal work of \cite{sims_1980}, and the so-called \emph{local projections} (LPs), introduced by \cite{jorda_2005}. It is now understood that there is a sense in which VARs and LPs share the same estimand \citep{plagborgmoller_wolf_2021,xu2026local}, but that in finite samples these procedures lie on opposite ends of a bias-variance trade-off \citep{li2024local}: LP estimators have lower bias than VAR estimators, but they also have substantially higher variance at intermediate and long horizons. 

The question of interest is how to best aggregate the information in LP and VAR estimators to estimate a particular impulse response coefficient. We follow closely the work of \cite{armstrong_kline_sun_2025} and consider the problem of constructing a minimax estimator based on quadratic loss, and a bivariate Gaussian model for the VAR and LP estimators.\footnote{In the bivariate Gaussian model, the LP estimator is unbiased but has potentially high variance, and the VAR estimator is biased but with smaller variance than the LP estimator; see Equation \eqref{eq:bivariate_normal_model}.} We impose an a priori bound on the bias of the VAR. As we explain in Section \ref{sec: LP-VAR}, both the bivariate normality and the bound on the bias can be motivated using the \emph{locally misspecified}
vector autoregression model used in \cite{montielolea_plagborgmoller_qian_wolf_forthcoming}. In particular, the bound on the bias can be obtained from restricting the fraction of the residual variance that is explained by the lagged shocks in the misspecified part of the VAR model; see \cite{montielolea_qian_wolf_plagborgmoller_2025}, p. 16. 

Under a convenient invariance assumption, the minimax invariant estimator is a \emph{debiased} version of the VAR estimator. The optimal debiasing function depends on the data only through the standardized difference between the VAR and LP estimators, which can be interpreted as a \cite{hausman1978specification} statistic for misspecification and which we denote as $\Delta$. As shown in \cite{armstrong_kline_sun_2025}, the optimal debiasing function turns out to correspond to the minimax estimator of a Bounded Normal Mean problem, where we use $\Delta$ to try to estimate the magnitude of the bias of the VAR estimator.  

We report our estimator for four influential macroeconomic papers, where we allow the fraction of the residual variance that is explained by the lagged shocks in the misspecified part of the VAR model to be at most 1\%. The four papers we consider are  \citet{gertler_karadi_2015}, focusing on the impulse response of log industrial production to a monetary policy shock; \citet{romer_romer_2010}, focusing on the impulse response of log of real GDP per capita to a tax shock; \citet{ramey_2011}, focusing on the impulse response of log of real GDP per capita to a government spending shock (based on military news); and \citet{francis_owyang_roush_dicecio_2014}, focusing on the impulse response of log of labor productivity to a technology shock. Even though the amount of misspecification that we allow is relatively small (which should favor the VAR estimator), in all applications the approximately minimax estimator is very close to the LP estimator. This is also consistent with the minimax linear estimator reported by \cite{montielolea_plagborgmoller_qian_wolf_forthcoming} (Corollary 4.2), which under our assumptions, places about 70\% of the weight on the LP estimator. In our applications, for medium and long horizons the \emph{shrinkage} towards the LP estimator is a bit more aggressive than the one implied by the linear estimator (and the gains in excess risk relative to an oracle that knows the bias are about 15\%). 

The remainder of the paper is organized as follows. Section \ref{sec: BNM} introduces the Bounded Normal Mean problem and its minimax formulation. Section \ref{sec: main_results} presents the algorithm and the main theoretical results. Section \ref{sec: simulation} reports simulation results. Section \ref{sec: LP-VAR} applies the method to estimation of impulse responses.

\section{The Bounded Normal Mean problem} \label{sec: BNM}

Fix $m > 0$ and consider the statistical model 
\[ Y \sim \mathcal N(\theta, 1), \quad \theta \in [-m,m].\]
The problem of interest is to estimate $\theta$ under the loss function: 
\[\mathcal{L} (a, \theta) =  (a- \theta)^2.\] 
\citet{casella_strawderman_1981} refer to this problem as the \emph{bounded normal mean estimation problem}. An excellent and detailed treatment of the problem is given in Chapter 4.6 of \cite{johnstone_2019}. 

A large literature in statistics has studied the theoretical properties of different estimators for the Bounded Normal Mean problem; see, for example, \cite{bickel1981minimax}, \cite{levit1981asymptotic}, \cite{donoho_liu_macgibbon_1990}, and the references therein. As we will explain below, an important difference of our paper relative to existing literature is that we make it our goal to suggest a concrete numerical strategy---based on the methods of \emph{stochastic mirror descent (ascent)}---to obtain an \emph{approximate} minimax estimator (in a sense we make precise). The main theoretical results in our paper will be focused on showing that our numerical strategy indeed achieves the desired objective. 

We note that we have deliberately set the variance of the outcome $Y$ to equal one. As we discuss later, this is without loss of generality as the Bounded Normal Mean problem treats the variance as known. See also the scale-invariance relation in Equation 4.35, Chapter 4.6 of \cite{johnstone_2019}. It is also without loss of generality to restrict the action space to be $[-m,m]$. The reason is that, under squared-error loss, truncating any real-valued estimator to $[-m,m]$ weakly reduces risk pointwise in $\theta$.\footnote{One benchmark is the projection of $Y$ onto $[-m,m]$,
$d_{\textrm{trunc}}(y) = \min \{\max \{y,-m\}, m\},$ which is measurable, valued in $[-m,m]$, and, by truncation, weakly dominates the unconstrained estimator $d(y) = y$ under squared loss as $(d_{\textrm{trunc}}(y) - \theta)^2 \le (y-\theta)^2$ for all $\theta \in \Theta$ pointwise.}

Let $\mathcal{D}$ be the class of all measurable decision rules $d: \mathbb{R} \to [-m,m]$. As usual, the \emph{risk} of a rule $d$ is
\[R (d, \theta) := \mathbb{E}_{\theta} [(d(Y)- \theta)^2],\] 
where the expectation is taken assuming $Y \sim \mathcal{N}(\theta,1)$. The first result of this paper establishes an important property of the risk function that will be exploited by our numerical strategy; namely, that  $R(d,\cdot)$ is \emph{Lipschitz continuous} in $\theta$ for any $d \in \mathcal{D}$. 
\begin{lemma} [The Risk Function is Lipschitz in $\theta$] \label{lemma:RiskFunction is Lips}
For any $d \in \mathcal{D}$ and any $\theta,\theta' \in [-m,m]$
\begin{equation} \label{eqn:RiskFunction is Lips} 
\left | R(d,\theta) - R(d,\theta') \right|   \leq 
\left(3m^2+ 4m \right) |\theta-\theta'|. 
\end{equation}
\end{lemma}
\begin{proof}
See Appendix \ref{subsec:Proof of Lemma 1}.
\end{proof}

This first lemma shows that the Lipschitz constant is independent of $d$, but dependent on the bound $m$. To the best of our knowledge this result has not appeared previously in the literature. The key argument we use to establish this result is an upper bound on the total variation distance between two normal distributions with unit variance; see Theorem 1.3 in \cite{devroye2018total}.\footnote{For normal distributions with common covariance, Theorem 1 in \cite{barsov1987proximity} implies the sharper bound: \[ \frac{1}{2} \int_{-\infty}^{\infty} \left| \phi(y;\theta) - \phi(y;\theta') \right| dy \le \frac{ |\theta-\theta'| }{\sqrt{2\pi}}.\] Using this bound in Appendix \ref{subsec:Proof of Lemma 1} would give the sharper Lipschitz constant equal to $3 \sqrt{\frac{2}{\pi}} m^2 + 4 m$.}

\subsection{Approximate Minimax Estimation in the Bounded Normal Mean problem} 

Given $m>0$, we say that an estimator $d^*$ is minimax---or \emph{$m$-minimax},  adopting the terminology of \cite{armstrong_kline_sun_2025}---if it satisfies
\begin{equation} \label{eqn:minimax}
\sup_{\theta \in [-m,m]} R(d^*,\theta) =  \bar{v}^*(m) := \adjustlimits \inf_{d \in \mathcal{D}} \sup_{\theta \in [-m,m]} R(d,\theta). 
\end{equation} 
We refer to $\bar{v}^*(m)$ as the \emph{minimax value} of the bounded normal mean estimation problem associated with the bound $m>0$. 

It is known that when the bound $m$ is sufficiently small, there is a closed-form solution for the problem in \eqref{eqn:minimax}; see Theorems 3.1 and 4.2 in \citet{casella_strawderman_1981}. More generally, the problem of finding a minimax estimator can be viewed as an infinite dimensional convex programming problem; see pp. 1419 and 1420 in \cite{donoho_liu_macgibbon_1990}. This means that finding a minimax estimator must necessarily involve some type of numerical approximation.  

To be explicit about the fact that any given numerical strategy will only approximately solve the minimax problem in \eqref{eqn:minimax}, we build on the seminal work of \cite{ghosh1964uniform} and define a \emph{stochastic minimax sequence} as follows. Let $\{P_n\}_{n=1}^{\infty}$ be a sequence of probability distributions over the space $\mathcal{D}$ and let $\{d^*_{n}\}_{n=1}^{\infty}$ be a stochastic sequence of decision rules in which each decision rule $d^*_n$ is drawn from $P_n$.

\begin{definition}[Stochastic Minimax Sequence] \label{def:stochastic_minimax_sequence} A stochastic sequence of decision rules $\{d^*_{n}\}_{n=1}^{\infty}$ is said to be a stochastic minimax sequence for the Bounded Normal Mean problem if for any $\varepsilon>0$
\begin{equation} \label{eq:stochastic_minimax_sequence}  
P_n \left( \sup_{\theta \in [-m,m]} R(d^*_n,\theta) \leq \bar{v}^*(m) + \varepsilon \right) \rightarrow 1,
\end{equation}
as $n \rightarrow \infty$. 
\end{definition} 
\noindent Definition \ref{def:stochastic_minimax_sequence} extends the classical definition of \emph{minimax sequence} given in \cite{ghosh1964uniform}, p. 1037. \cite{ghosh1964uniform} defines a deterministic sequence of decision rules $\{d_n\}_{n=1}^{\infty} \subseteq \mathcal{D}$ to be a minimax sequence if 
\begin{equation} \label{eq:Ghosh's definition}
\sup_{\theta \in [-m,m]} R(d_n,\theta) \rightarrow \bar{v}^*(m). 
\end{equation}
Since the numerical procedure that we will recommend (which is akin to stochastic gradient descent for convex minimization problems) introduces noise beyond the statistical uncertainty contained in the estimation problem, we need to allow for the possibility that $d_n$ is stochastic. To this end, Definition \ref{def:stochastic_minimax_sequence} replaces the $``\rightarrow''$ relation in \eqref{eq:Ghosh's definition} by convergence in probability, where convergence in probability is assessed relative to the distributions $P_n$.  

Definition \ref{def:stochastic_minimax_sequence} is also related to the notion of $\epsilon$-minimax decision rules, which is a common way of formalizing the notion of an approximate minimax solution in statistical decision problems; see \cite{ferguson_1967}, Chapter 1, Definition
4, p. 33.\footnote{See also \cite{montielolea_aradillasfernandez_blanchet_qiu_stoye_tan_2024} for recent applications and an algorithmic approach for finding $\epsilon$-minimax decision rules in a class of statistical decision problems.} Note that if $\{d^*_n\}_{n=1}^{\infty}$ is a stochastic minimax sequence, then for every pair $\epsilon,\delta > 0$ there exists $n(\epsilon,\delta)$ such that whenever $n \geq n(\epsilon,\delta)$
\[ P_n \left( \bar{v}^*(m) \leq \sup_{\theta \in [-m,m]} R(d^*_n,\theta) \leq \bar{v}^*(m) + \epsilon \right) \geq 1-\delta. \]
This means that the \emph{tail} of a stochastic minimax sequence contains decision rules that are $\epsilon$-minimax with probability at least $1-\delta$.

\subsection{Approximate Least-Favorable Distribution for a Discretized Version of the Bounded Normal Mean problem}

Our strategy for constructing a stochastic minimax sequence for the problem in \eqref{eqn:minimax} consists of finding a sequence of priors that approximate the problem's \emph{least-favorable distribution}. We do this by discretizing the parameter space. The strategy of discretizing the parameter space to find an approximate least-favorable distribution has been recommended before in the econometrics literature by \cite{chamberlain2000econometric} and it has been used recently by \cite{armstrong_kline_sun_2025}. To the best of our knowledge, neither paper provides theoretical results that guarantee that such a strategy indeed succeeds in finding an approximate minimax rule (or even an approximate least-favorable distribution). The closest result that we are aware of is Lemma 1 of \cite{guggenberger2025numerical}, which provides a set of sufficient conditions that guarantee that the value of a discretized version of a minimax problem converges to the minimax value of the original problem.

For the sake of exposition, this section presents a quick overview of the typical relation between a minimax decision rule and a least-favorable distribution. We also present the discretized Bounded Normal Mean problem and the definition of an approximate least-favorable distribution for this problem.  

Let $\tau$ denote the standard (subspace) topology over $[-m,m]$. It is known that the topological space $([-m,m],\tau)$ is a Polish space.\footnote{This follows from the fact that $\tau$ is metrizable by the absolute value metric, and the space $([-m,m],|\cdot|)$ is a complete and separable metric space.} In order to define probability distributions over $[-m,m]$, we endow this space with the $\sigma$-algebra generated by $\tau$; i.e., its Borel $\sigma$-algebra. Let $\Delta([-m,m])$ denote the set of all (Borel) probability measures over $[-m,m]$.  

Let $\pi$ be an element of $\Delta([-m,m])$. The \emph{Bayes risk} of $d$ with respect to $\pi$ is defined as
\[r(d, \pi) := \mathbb{E}_{\pi}[\mathbb{E}_{\theta} (d(Y)- \theta)^2] = \int_{[-m,m]} R(d,\theta)  d\pi.\]
The rule $d_{\pi}$ is said to be a Bayes rule with respect to $\pi$ if it satisfies 
\begin{equation}\label{eqn:Bayes}
r(d_{\pi},\pi) = \inf_{d \in \mathcal{D}} r(d,\pi). 
\end{equation}
Define a \emph{least-favorable distribution} as a distribution $\pi^*$ that satisfies
\begin{equation}\label{equation:maximin}
\inf_{d\in \mathcal{D}} r(d,\pi^*) =  \underline{v}^* (m) := \sup_{\pi\in\Delta([-m,m])} \inf_{d \in \mathcal{D}} r(d,\pi).
\end{equation} 
In Appendix \ref{section: math_properties} we show that the maximin problem that defines the least-favorable distribution can be viewed as a concave maximization problem over the infinite-dimensional space $\Delta([-m,m])$. The objective function in this problem---which corresponds to the function $f(\pi) = \inf_{d \in \mathcal{D}} r(d,\pi)$---is Lipschitz continuous in $\pi$, with respect to the 1-Wasserstein metric over $\Delta([-m,m])$. 

It is known that  
\begin{equation} \label{eqn:minimax_theorem}
\bar{v}^*(m)=\underline{v}^*(m).
\end{equation}
See Proposition 4.19 in \cite{johnstone_2019} and also their Theorem A.5.\footnote{There are two additional interesting properties that we will reference later. First, the least-favorable distribution is unique. Second, the least-favorable distribution is supported on finitely many points (and all support points have the same risk). See, again, Proposition 4.19 in \cite{johnstone_2019}.} It is also known that if \eqref{eqn:minimax_theorem} holds, any minimax rule $d^*$ must be a Bayes rule with respect to the least-favorable distribution $\pi^*$.

\emph{Discretized Maximin Problem:} As mentioned before, our numerical strategy to construct a stochastic minimax sequence starts by discretizing the parameter space $\Theta := [-m, m]$ using an equally-spaced grid of $I>1$ points $\{\theta_1,\dots,\theta_{I}\}$ in the interval $[-m,m]$. This means that the $i$-th point in the grid equals:
\[ \theta_i: = -m +  \frac{2m(i-1)}{I-1}. \]
The gap between any consecutive points $\theta_{i+1}, \theta_i$ is 
\begin{equation} \label{eqn: gap Delta}
\mathbf{\Delta}(m,I) := \theta_{i+1} - \theta_i = \frac{2m}{I-1}.  
\end{equation}

Define the $(I-1)$-dimensional simplex
$\Delta^{I-1} := \left \{\pi \in [0,1]^I: \sum_i \pi_i = 1 \right\}$ and then consider the \emph{discretized} version of the maximin problem in \eqref{equation:maximin}
\begin{equation} \label{eq:discretized maximin problem}
\underline{v}(m;I) := \max_{\pi\in\Delta^{I-1}} \inf_{d \in \mathcal{D}} r(d,\pi), \quad \textrm{where } \quad r(d,\pi) := \sum_{i=1}^{I}\pi_i R(d,\theta_i).
\end{equation}

There is an evident gap between the value of the original maximin problem in \eqref{equation:maximin}---which we denoted as $\underline{v}^*(m)$---and the value of the discretized maximin problem in \eqref{eq:discretized maximin problem}---which we are denoting as $\underline{v}(m;I)$. Obviously, $\underline{v}(m;I) \leq \underline{v}^*(m)$, since the feasible set of the outer maximization in the discretized maximin problem is a subset of $\Delta([-m,m])$. The second result of this paper establishes a concrete quantitative relation between the maximin values of the Bounded Normal Mean problem and its discretized version.   

\begin{lemma} [Distance between the original and discretized maximin problems] \label{lemma:discretized vs original} 
For any $m>0$ and $I>1$ 
\[ \underline{v}(m;I) \leq \underline{v}^*(m)    
\leq \underline{v}(m;I)  + (3m^3+4m^2)/(I-1).\]
\end{lemma}
\begin{proof}
See Appendix \ref{subsec:Proof of Lemma 2}.
\end{proof}

The key argument to establish Lemma \ref{lemma:discretized vs original} is the Lipschitz continuity of the risk function (which was established in Lemma \ref{lemma:RiskFunction is Lips}) and the minimax theorem applied to both the discretized and the original problem. As expected, the gap between the discretized maximin problem and the original problem decreases as a function of $I$, and larger values of $m$ will require a finer grid to attain a good approximation.  We note that Lemma \ref{lemma:discretized vs original} readily yields a formula to select $I$ in order to guarantee that the gap between the two maximin problems is small. For example, to guarantee a gap of at most $\epsilon/2$ it suffices to set the number of points in the uniform grid equal to
\begin{equation} \label{eqn:I of epsilon}
I(\epsilon) : = \left \lceil 1 + \frac{3m^3 + 4m^2}{\epsilon/2} \right \rceil.
\end{equation}
The function $\lceil \cdot \rceil$ is the ceiling function: the function that returns the smallest integer that is greater than or equal to a given number. In light of Lemma \ref{lemma:discretized vs original} our numerical strategy will focus on obtaining an approximate solution to the discretized maximin problem. The following definition is entirely analogous to the notion of  $\epsilon$-minimax decision rules---see \cite{ferguson_1967}, Chapter 1, Definition 4, p. 33---but applied to the maximin problem.

\begin{definition}[$\epsilon$-least-favorable distribution for the discretized maximin problem] We say that $\pi^{(\epsilon,I)} \in \Delta^{I-1}$ is an $\epsilon$-least-favorable distribution for the discretized maximin problem in \eqref{eq:discretized maximin problem} if
\[  \underline{v}(m;I) - \epsilon \leq \inf_{d \in \mathcal{D}} r(d,\pi^{(\epsilon,I)}) \leq  \underline{v}(m;I).   \]
We say that a distribution $\pi^{I} \in \Delta^{I-1}$ is an approximately least-favorable distribution for the discretized maximin problem in \eqref{eq:discretized maximin problem}, if there exists $0<\epsilon<4m^2$ for which $\pi^I$ is an $\epsilon$-least-favorable distribution.
\end{definition}

Let $I(\epsilon)$ be defined as in \eqref{eqn:I of epsilon}. Note that an important implication of Lemma \ref{lemma:discretized vs original} is that if $\pi^{(\epsilon/2,I(\epsilon))}$ is an $(\epsilon/2)$-least-favorable distribution for the discretized problem with $I(\epsilon)$ equally-spaced points, then $\pi^{(\epsilon/2,I(\epsilon))}$ is an $\epsilon$-least-favorable distribution of the \emph{original} Bounded Normal Mean problem. Our strategy to construct a stochastic minimax sequence for the Bounded Normal Mean problem relies on this observation. More concretely, we will consider a sequence $\{(\epsilon_n,I_n)\}_{n=1}^{\infty}$, where $\epsilon_n \rightarrow0$ and $I_n \rightarrow \infty$. We will show that we can build an associated stochastic sequence of priors $\pi_n \equiv \pi^{\epsilon_n}$ that are guaranteed to be---with high probability---approximately the least-favorable distribution for the original problem. In the next section we will show how to obtain such a stochastic sequence of priors. Our main result will be that the sequence of Bayes decision rules corresponding to the sequence of priors $\{\pi_n\}_{n=1}^{\infty}$ is a stochastic minimax sequence.

Note also that without an upper bound on $\epsilon$, any distribution $\pi \in \Delta^{I-1}$ is approximately least-favorable. The risk function in the Bounded Normal Mean problem is bounded by $4m^2$, so we are only interested in considering approximation errors below this upper bound.

\section{Main Results} \label{sec: main_results}

This section presents our main results. First, we show that a \emph{stochastic mirror ascent} algorithm for concave maximization over the $(I(\epsilon)-1)$-dimensional simplex can be used to provably obtain---with high probability---an approximate least-favorable distribution for the maximin problem in \eqref{equation:maximin}. Second, we show that the Bayes rules associated with the sequence of approximately least-favorable priors obtained via stochastic mirror ascent constitute a stochastic minimax sequence.

\subsection{Stochastic Mirror Ascent Algorithm for the Discretized Maximin Problem} 
For any $\pi \in \Delta([-m,m])$, define the function $f: \Delta([-m,m]) \rightarrow \mathbb{R}_{+}$ given by
\begin{equation*}
f(\pi) := \inf_{d \in \mathcal{D}} r(d,\pi) = \inf_{d \in \mathcal{D}} \int_{[-m,m]} R(d,\theta) d\pi .
\end{equation*}
The first key observation here is that the function $f(\cdot)$ is concave in $\pi$ over $\Delta([-m,m])$, and also over $\Delta^{I-1}$ for any $I > 1$. This result has been established more generally before; see, for example,  \cite{chamberlain2000econometric}, Equation 5, p. 630, and the discussion therein. For the sake of completeness, we formalize this observation in Lemma \ref{lemma:OptimizedBayesRule is concave} in Appendix \ref{section: math_properties}. Lemma \ref{lemma:OptimizedBayesRule is concave} also shows that for any prior $\pi \in \Delta^{I-1}$, the vector $g(\pi):=(R(d_{\pi}, \theta_1), \ldots, R(d_{\pi},\theta_{I}))$ is a supergradient of $f(\cdot)$ at $\pi$ (where $d_{\pi}$ is the Bayes rule corresponding to $\pi$). Lastly, since $|R(d, \theta)| \leq 4m^2$ for all $(d,\theta) \in \mathcal{D} \times \Theta$, the function $f:\Delta^{I-1} \rightarrow \mathbb{R}$ is Lipschitz (with respect to the $L_1$-norm) with Lipschitz constant $M:=4m^2$.

A popular algorithmic approach for \emph{approximately} solving concave maximization problems---in particular, those of high dimensions---is to use the methods of (stochastic) mirror descent of  \cite{nemirovski_yudin_1983}. The methods of mirror descent are a family of iterative procedures designed for finding an  \emph{approximate solution} of convex minimization (or concave maximization) problems in high dimensions. These methods require repeated (but finitely many) evaluations of the subgradient of the objective function or a noisy but unbiased estimator of the subgradient. These methods exploit the geometry of the optimization domain, but they do not exploit any higher-order smoothness information about the optimization problem. Mirror descent is known to outperform gradient descent in some optimization domains, such as the probability simplex; see Section 4.3 of \cite{bubeck_2015}. See also Section 5.1.1 in \cite{nemirovski2009robust}.  

We consider the iterative algorithm described below for maximizing $f(\cdot)$ over the simplex $\Delta^{I-1}$. The stochastic component of this algorithm comes from the evaluation of the supergradient. The pseudo-code that describes Algorithm \ref{alg:SMA} below takes the grid $\{\theta_i\}_{i=1}^{I}$ as given. Another input is the \emph{step size} (denoted as $\eta$), which controls the adjustment of the prior at each epoch. We also take the number of epochs ($J$) as given. 

\begin{algorithm}[h]
\caption{Stochastic Mirror Ascent to Find Approximate Least-Favorable Prior}\label{alg:SMA}
\begin{algorithmic}[1]
\State \textbf{Input:} Step size $\eta > 0$; number of epochs $J \in \mathbb{N}$; equally-spaced grid $\{\theta_i\}_{i=1}^I \subset \Theta$.
\State Initialize $w_0 \in \mathbb{R}^I$ by setting $w_{i,0} = 1$ for all $i \in \{1,\ldots,I\}$.
\For{$j = 1, 2, \ldots, J$}
    \State Compute $\phi_j := \sum_{i=1}^{I} w_{i,j-1}$.
    \State For each $i \in \{1,\ldots,I\}$, compute
        \[\pi_{i,j} := \frac{w_{i,j-1}}{\phi_j}.\]
    \State Compute the Bayes decision rule $d_{\pi_j}$ under prior $\pi_j$:\footnote{Since $\pi_j \in \Delta^{I-1}$, the Bayes rule $d_{\pi_j}$ admits this closed-form posterior mean representation, which can be evaluated pointwise for any $y$.}
\[d_{\pi_j}(y) = \mathbb{E}_{\pi_j}[\theta \mid y] = \frac{\sum_{i=1}^{I} \theta_i \pi_{i,j} \phi(y;\theta_i)}{\sum_{i=1}^{I} \pi_{i,j} \phi(y;\theta_i)},\]
where $\phi(y;\theta_i)$ is the p.d.f.\ of $\mathcal{N}(\theta_i,1)$.
    \State For each $i \in \{1,\ldots,I\}$, draw an independent $y_{i,j} \sim \mathcal{N}(\theta_i,1)$ and compute the stochastic supergradient estimate
        \[\tilde{g}_{i,j} = \bigl(d_{\pi_j}(y_{i,j}) - \theta_i\bigr)^2.\]
    \State Multiplicative weights update: for each $i \in \{1,\ldots,I\}$,
\[w_{i,j} := w_{i, j-1} \cdot \exp\bigl(\eta \tilde{g}_{i,j} \bigr).\]
\EndFor
\State \textbf{Output:} $\frac{1}{J} \sum_{j=1}^{J}\pi_j = \left(\frac{1}{J}\sum_{j=1}^{J}\pi_{1,j}, \dots, \frac{1}{J}\sum_{j=1}^{J}\pi_{I,j} \right).$
\end{algorithmic}
\end{algorithm}

The next theorem shows that---with high probability---the stochastic mirror ascent routine outputs an approximate least-favorable distribution for the Bounded Normal Mean problem. 

\begin{theorem}[Approximately least-favorable distribution]
\label{thm:avg_prior_pointwise_LF}
Fix $\epsilon>0$ and $m>0$. Let $\pi^{\epsilon}$ be the output of Algorithm \ref{alg:SMA} with inputs
\begin{equation} \label{eq: tuning_params}
    \eta(\epsilon) := \frac{\epsilon/2}{M^2}, \qquad J(\epsilon) := \left \lceil \frac{2 M^2 \ln (I(\epsilon))}{(\epsilon/2)^2} \right \rceil, \quad I(\epsilon):= \left \lceil 1 + \frac{3m^3 + 4m^2}{\epsilon/2} \right \rceil,
\end{equation}
where $M:=4m^2$. Then, for any $\alpha \in (0,1)$, with probability at least $1-\alpha$
we have
\begin{equation}\label{eq:avg_prior_pointwise_LF}
\inf_{d\in\mathcal{D}} r(d,\pi^\epsilon) \ge \underline{v}^*(m)-\epsilon \left(1 + \sqrt{\frac{\ln(1/\alpha)} {\ln(I(\epsilon))}} \right).
\end{equation}
\end{theorem}

\begin{proof}
See Appendix \ref{subsec:Proof of Theorem 1}. 
\end{proof}

\noindent Theorem \ref{thm:avg_prior_pointwise_LF} shows that our suggested stochastic mirror ascent algorithm---with the choice of tuning parameters in \eqref{eq: tuning_params}---succeeds in finding an approximate least-favorable distribution for the maximin problem \eqref{equation:maximin}. Since Algorithm \ref{alg:SMA} relies on an estimator of the supergradient (see Step 7), it is important to quantify this uncertainty in our theoretical guarantees. The theorem provides an explicit description of the effects of the stochastic gradient on the optimization procedure. In particular, the theorem says that with probability at least $1-\alpha$, the Bayes risk of $\pi^\epsilon$ is at most within the additive factor 
\[ \epsilon \left(1 + \sqrt{\frac{\ln(1/\alpha)} {\ln(I(\epsilon))}} \right)\]
of the maximin value of interest. Note that this term is increasing in $\epsilon$. However, a smaller value of $\epsilon$ increases the computational complexity of the procedure because it increases the dimension of the concave optimization problem via its effect on $I(\epsilon)$.

To complement Theorem \ref{thm:avg_prior_pointwise_LF} we also show that there is a sense in which $\pi^{\epsilon}$ is close to the (unique) least-favorable distribution of the original maximin problem, which we have denoted as $\pi^*$. To show this, take an arbitrary sequence $\{\epsilon_n\}_{n=1}^{\infty} \subseteq \mathbb{R}_{++}$ such that $\epsilon_n  \rightarrow 0$. For each $\epsilon_n$ in the sequence, we obtain the corresponding approximate least-favorable prior $\pi_n : = \pi^{\epsilon_n}$ as described in Theorem \ref{thm:avg_prior_pointwise_LF}. We show that $\pi_n$ converges (in 1-Wasserstein distance, as defined below) to the least-favorable distribution $\pi^*$. 

\noindent \emph{The 1-Wasserstein Distance.} For arbitrary Borel probability distributions $\pi, \pi' \in \Delta([-m,m])$ define $\Gamma(\pi,\pi')$ to be the set of all couplings of $(\pi,\pi')$. Define the 1-Wasserstein distance (with cost function $|\cdot|$) as
\[ W(\pi,\pi') := \inf_{\gamma \in \Gamma(\pi,\pi')} \int_{[-m,m]^2} |\theta-\theta'|d\gamma. \]
It is known that $W(\pi,\pi')$ defines a metric over all Borel probability measures on $\Delta([-m,m])$; see Theorem 7.3 in \cite{villani2021topics}. Moreover, it is common to refer to $W(\cdot,\cdot)$ as the \emph{Kantorovich-Rubinstein} distance. It is also known that the Kantorovich-Rubinstein duality formula implies that 
\begin{equation} \label{eqn:KR-duality} 
W(\pi,\pi') = \sup_{\varphi} \left( \int_{[-m,m]} \varphi d\pi - \int_{[-m,m]} \varphi d\pi' \right) \quad \textrm{s.t.} \quad  \| \varphi \|_{\textrm{Lip}} \leq 1, 
 \end{equation}
where  
\[ \| \varphi \|_{\textrm{Lip}} := \sup_{\theta \neq \theta'} \frac{|\varphi(\theta)-\varphi(\theta')|}{|\theta-\theta'|}. \]
See Remark 7.5, p. 207 in \cite{villani2021topics}.

\begin{theorem}[$\pi_n$ converges to $\pi^*$ in 1-Wasserstein Distance]  \label{thm:convergence of pi_n} Fix $m>0$. Take an arbitrary sequence $\{\epsilon_n\}_{n=1}^{\infty} \subseteq \mathbb{R}_{++}$ such that $\epsilon_n  \rightarrow 0$. Let $\pi_n:=\pi^{\epsilon_n}$ be the output of Algorithm \ref{alg:SMA} with the parameters specified in Theorem \ref{thm:avg_prior_pointwise_LF}. Let $\pi^*$ be the unique least-favorable prior of the Bounded Normal Mean problem. Then, for any $\xi>0$
\[ P_n \left( W(\pi_n,\pi^*) > \xi \right) \rightarrow 0, \textrm{ as } n \rightarrow \infty.  \]
\end{theorem}
\begin{proof}
See Appendix \ref{subsec:Proof of Theorem 2}.
\end{proof}

\subsection{Stochastic Minimax Sequence for the Bounded Normal Mean problem}
Theorems \ref{thm:avg_prior_pointwise_LF} and \ref{thm:convergence of pi_n} have shown that the stochastic mirror ascent algorithm can be used to generate an approximate least-favorable prior. In this section, we use these results to construct a stochastic minimax sequence. 

\begin{theorem} \label{thm:minimax sequence} Fix $m>0$. Take an arbitrary sequence $\{\epsilon_n\}_{n=1}^{\infty} \subseteq \mathbb{R}_{++}$ such that $\epsilon_n  \rightarrow 0$. Let $\pi_n:=\pi^{\epsilon_n}$ be the output of Algorithm \ref{alg:SMA} with the parameters specified in Theorem \ref{thm:avg_prior_pointwise_LF}. Let $d_n$ denote the Bayes rule associated with $\pi_n$. The sequence $\{d_n\}_{n=1}^{\infty}$ is a stochastic minimax sequence. 
\end{theorem}

\begin{proof}
See Appendix \ref{subsec:Proof of Theorem 3}.  
\end{proof}

\begin{theorem} \label{thm:rule_convergence}
Fix $m>0$ and let $S \subset \mathbb{R}$ be any compact subset. Under the same conditions as Theorem \ref{thm:minimax sequence}, for any $\xi > 0$
\[P_n \left(\sup_{y \in S} |d_n(y) - d_{\pi^*}(y)| > \xi \right) \to 0, \textrm{ as } n \rightarrow \infty. \]
\end{theorem}

\begin{proof}
See Appendix \ref{subsec:Proof of Theorem 4}.
\end{proof}

Theorems \ref{thm:minimax sequence} and \ref{thm:rule_convergence} translate the convergence of the approximate least-favorable distribution into convergence of the approximately minimax estimator. Theorem \ref{thm:minimax sequence} shows that the sequence of Bayes rules corresponding to the approximately least-favorable distribution forms a bona fide stochastic minimax sequence. The stochastic convergence reflects the randomness from the Monte Carlo estimation of supergradients. Theorem \ref{thm:rule_convergence} complements Theorem \ref{thm:minimax sequence} by showing that the sequence of approximately minimax rules converges (uniformly over compacts) with high probability to the problem's minimax rule.

\section{Simulations} \label{sec: simulation}

This section explores the numerical performance of the stochastic mirror ascent algorithm in the Bounded Normal Mean problem. We compare the risk of the approximately minimax estimator to the linear minimax benchmark. We also assess the computational cost of the suggested procedure.

\subsection{Selection of the parameter $\epsilon$}
As discussed in the introduction---with the exception of $\epsilon$---our Theorem \ref{thm:avg_prior_pointwise_LF} makes explicit recommendations about all of the parameters that need to be specified in order to run the stochastic mirror ascent algorithm. We now argue that in the Bounded Normal Mean problem, there is a simple way to choose $\epsilon$ (as a function of $m$).  

Note first that the minimax risk over linear estimators in the Bounded Normal Mean problem is 
\begin{equation} \label{eqn:linear minimax} 
\bar{v}_{\mathrm{linear}}(m) := \frac{m^2}{1+m^2},
\end{equation}
see Equation 4.28 in \cite{johnstone_2019}. It is also known that
\begin{equation} \label{eqn:ibragimov} \bar{v}^*(m) \leq \bar{v}_{\mathrm{linear}}(m) \leq 1.25 \bar{v}^*(m), \end{equation} where the upper bound is due to \cite{donoho_liu_macgibbon_1990}; 
see also Theorem 4.17 in \cite{johnstone_2019}.\footnote{The factor 1.25 is typically referred to as the \cite{ibragimov1985nonparametric} constant.} 

If $\epsilon$ were chosen to satisfy $\epsilon > 0.25 \bar{v}^*(m)$, then the minimax linear estimator would be $\epsilon$-minimax:  $\bar{v}_{\mathrm{linear}}(m) \leq 1.25 \bar{v}^*(m) < \bar{v}^*(m) + \epsilon$. We are interested in values of $\epsilon$ associated with smaller approximation errors (to avoid having provable $\epsilon$-minimax guarantees for the minimax linear estimator). Thus, using \eqref{eqn:ibragimov} we recommend choosing $\epsilon$ as  
\begin{equation}\label{eqn:bound_epsilon} 
\epsilon(m) := \frac{1}{5} \frac{m^2}{1+m^2} = \frac{0.25}{1.25} \bar{v}_{\mathrm{linear}}(m) \leq 0.25 \bar{v}^*(m). 
\end{equation}

Figure \ref{fig: J_and_I_of_m} displays both the required number of iterations of the stochastic mirror ascent algorithm and the implied grid size when we choose $\epsilon(m)$ as in \eqref{eqn:bound_epsilon} and   $m \in [0.1,9]$. The figure illustrates that both the number of iterations and the grid size grow (polynomially) with $m$. The dominant cost is the number of iterations, since $J(m):=J(\epsilon(m))$ grows as $m^4\ln(m)$ while the number of grid points $I(m):=I(\epsilon(m))$ grows as $m^3$. For example, at $m=1$, $I(m) = 141$ and $J(m) \approx 63,345$ (yielding a total time of about 3 seconds for the algorithm to run); while at $m=3$, $I(m) = 1,301$ and $J(m) \approx 2.3$ million (which makes the runtime of the algorithm about 2 hours).

\begin{figure}[h!] 
\begin{center}	
\includegraphics[width=5in]{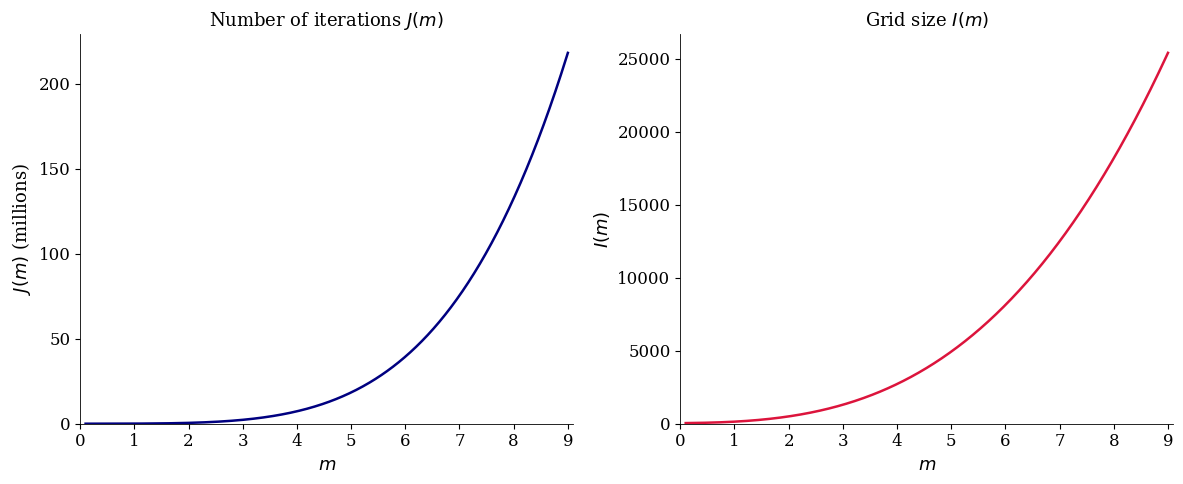}
\caption{Number of iterations and grid size as functions of $m$ for $m \in [0.1,9]$.}
\label{fig: J_and_I_of_m}
\end{center}	
\end{figure}

Figure \ref{fig: time_SMD_new} illustrates the total time required to run the stochastic mirror ascent routine and evaluate the risk with 1,000 MC draws and the values of $\epsilon(m)$ for $m \in \{1, \ldots, 3\}$, implied by the formula in \eqref{eqn:bound_epsilon}.

\begin{figure}[h!] 
\begin{center}	
\includegraphics[width=3.75in]{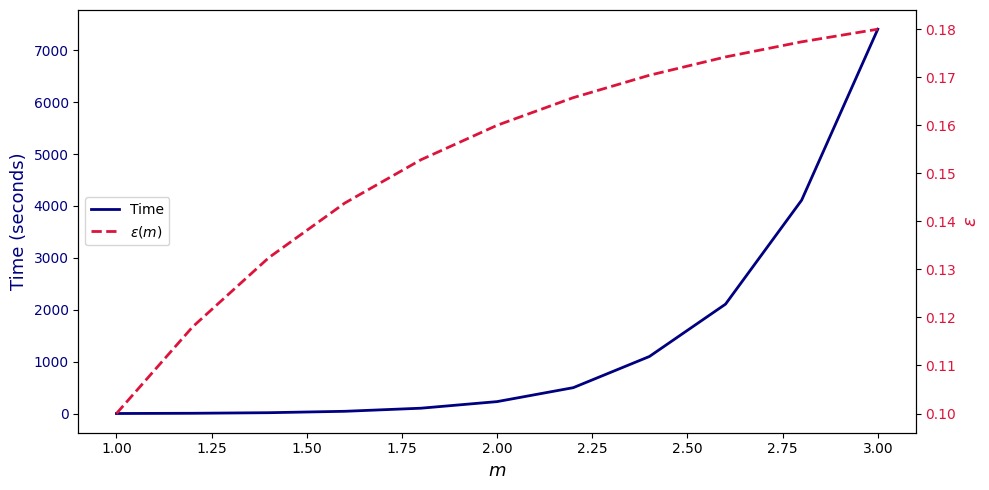}
\caption{Computation time and $\epsilon(m)$.}
\label{fig: time_SMD_new}
\end{center}	
\end{figure}

\subsection{Improvement of the approximately minimax estimator relative to the minimax linear estimator}

Figure \ref{fig: risk_all_m_4} compares the risk of the minimax estimator obtained via stochastic mirror ascent and the linear minimax estimator for $m \in \{1,\ldots,4\}$. The upper bound (gray dashed line) is $\bar{v}_{\mathrm{linear}}(m)=m^2/(1+m^2)$, the worst-case risk of the minimax linear estimator. The lower bound (orange dashed line) is $0.8 \bar{v}_{\mathrm{linear}}(m)$, obtained from the \cite{ibragimov1985nonparametric} constant. The difference between these upper and lower bounds equals $(1/5)m^2/(1+m^2)$, which is our choice of $\epsilon(m)$.

We obtain the worst-case risk of the estimator generated by the stochastic mirror ascent algorithm (red line, with circle markers) in two steps. In the first step, we evaluate the risk function of the approximately minimax estimator over a finite grid of 200,000 equally-spaced points in the interval $[-m,m]$. The risk function is evaluated by Monte Carlo simulation using $1,000$ draws.\footnote{We simulate $N=1,000$ independent draws $Z_{j} \sim N(0,1)$ and compute, for each grid point $\theta_i$,
\begin{equation}
\label{eq:mc_risk_sim}
\widehat{R}_N(d,\theta_i) = \frac{1}{N} \sum_{j=1}^{N} \left[ d(\theta_i+Z_{j})-\theta_i \right]^2.
\end{equation}} In the second step, we acknowledge the fact that the worst-case risk evaluated on the finite grid understates the estimator's worst-case risk over the parameter space (since the worst-case risk may be attained between grid points). In Appendix \ref{subsec:grid_vs_interval_risk} we show how the Lipschitz continuity of the risk function in Lemma \ref{lemma:RiskFunction is Lips} can be used to translate the evaluation of the risk over a finite grid into a \emph{certifiable} upper bound for the worst-case risk of the approximately minimax estimator---and, consequently, for the problem's minimax value. The red line (with circle markers) in Figure \ref{fig: risk_all_m_4} is such an upper bound. 

The blue line (with circle markers) is the average risk---computed according to the approximate least-favorable distribution---of the estimator generated by the stochastic mirror ascent algorithm. Since, by construction, the approximately minimax estimator is a Bayes decision rule with respect to the approximately least-favorable prior, the blue line (with circle markers) is a legitimate lower bound for both the problem's minimax value and for the estimator's worst-case risk.  

Since the worst-case risk of $d_{\pi^\epsilon}$---the estimator generated by the stochastic mirror ascent algorithm---lies in between the red and blue lines (both with circle markers), we conclude that $d_{\pi^\epsilon}$ is $\epsilon$-minimax optimal for a value of $\epsilon$ considerably smaller (the difference between the lines with circle markers) than the target $\epsilon(m)$ (the difference between the dashed lines). Figure \ref{fig: risk_all_m_4} is thus graphical evidence of the success of our stochastic mirror ascent algorithm.

\begin{figure}[h!] 
\begin{center}	
\includegraphics[width=3.75in]{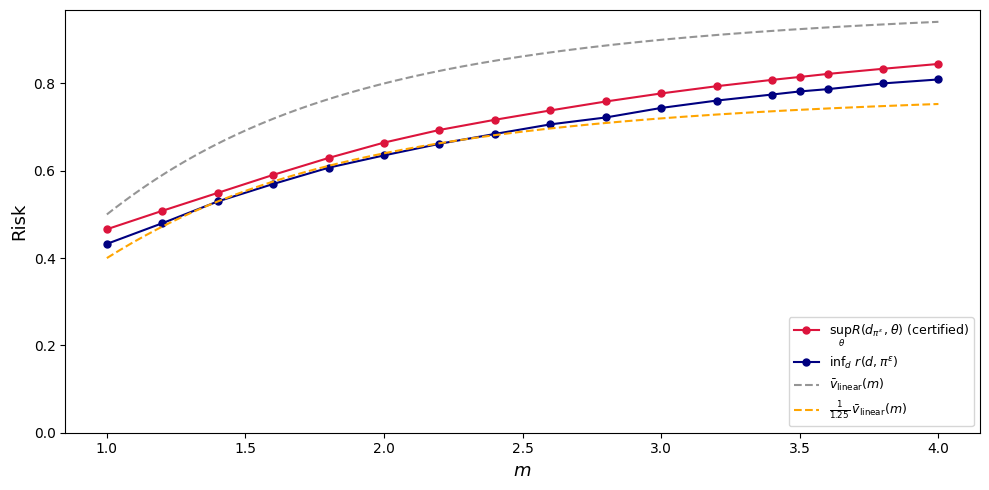}
\caption{Bounds on the worst-case risk of $d_{\pi^\epsilon}$.}
\label{fig: risk_all_m_4}
\end{center}	
\end{figure}

To further illustrate the success of our stochastic mirror ascent algorithm, Figure \ref{fig: improvement_all} reports a lower bound on the percentage improvement in the worst-case risk of $d_{\pi^\epsilon}$ relative to the worst-case risk of the minimax linear estimator. That is, we would like to report the risk improvement computed as
\[100 \times \left(1 - \frac{\sup_{\theta \in [-m,m]} R(d_{\pi^\epsilon}, \theta)}{\bar{v}_{\mathrm{linear}}(m)} \right),\]
but we replace $\sup_{\theta \in [-m,m]} R(d_{\pi^\epsilon}, \theta)$ by the certified upper bound reported in Figure \ref{fig: risk_all_m_4}.

\begin{figure}[h!] 
\begin{center}	
\includegraphics[width=3.75in]{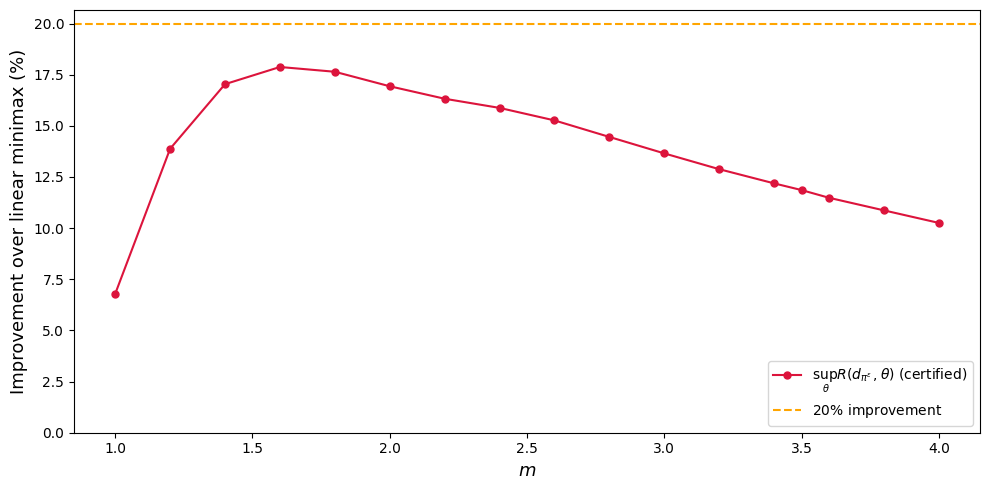}
\caption{Lower bound on percentage improvement in worst-case risk of $d_{\pi^\epsilon}$ over the worst-case risk of the minimax linear estimator.}
\label{fig: improvement_all}
\end{center}	
\end{figure}

For comparison, in Appendix \ref{subsec:grid_vs_interval_risk} we also report the worst-case risk of two alternative estimators: the truncated maximum likelihood estimator $d_{\textrm{trunc}} (y) = \max\{\min\{y,m\}, -m\}$ and the \emph{clipped} linear estimator (purple line) $d_{\textrm{linear, clipped}} = \max\{\min\{c \cdot y,m\}, -m\}$, where $c=m^2/(1+m^2)$.

\subsection{Comparison of $d_{\pi^\epsilon}$ and the minimax linear estimator}

Figure \ref{fig: estimators_m_1_2_4} plots the approximately minimax estimator and the minimax linear estimator for $m \in \{1,2,4\}$. Note that for $m=2$ and $m=4$, the most visible differences between $d_{\pi^\epsilon}$ and the minimax linear estimator concentrate near the boundaries.

\begin{figure}[h!] 
\begin{center}	
\includegraphics[width=5in]{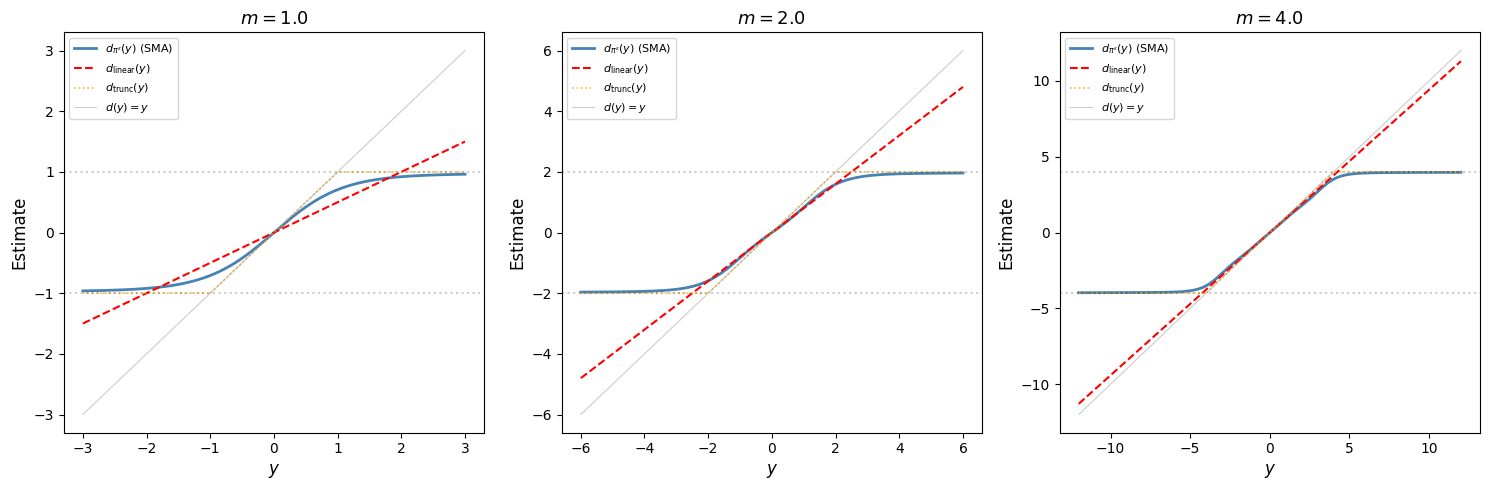}
\caption{Comparison of different estimators for the Bounded Normal Mean problem.}
\label{fig: estimators_m_1_2_4}
\end{center}	
\end{figure}

\subsection{Simulations with a time limit}

Suppose that we are only willing to run the algorithm for at most $x$ minutes. That is, we will run the algorithm either until the required number of iterations is reached or the time budget is exhausted, whichever happens first. These types of simulations are typically referred to as \emph{fixed wall-clock time} simulations.\footnote{See Section 7 of \cite{JMLR:v24:21-0377} for an example of numerical simulations that also fix the wall-clock time. } 

Figure \ref{fig: J_required_vs_completed_4} displays the number of iterations completed within a five-minute limit for each $m \in \{1, \dots, 4\}$, versus the theoretical target implied by $\epsilon(m) = \frac{1}{5} \frac{m^2}{1+m^2}$. As illustrated by this figure, the five-minute time limit starts binding at $m=2.2$.

\begin{figure}[h!] 
\begin{center}	
\includegraphics[width=3.75in]{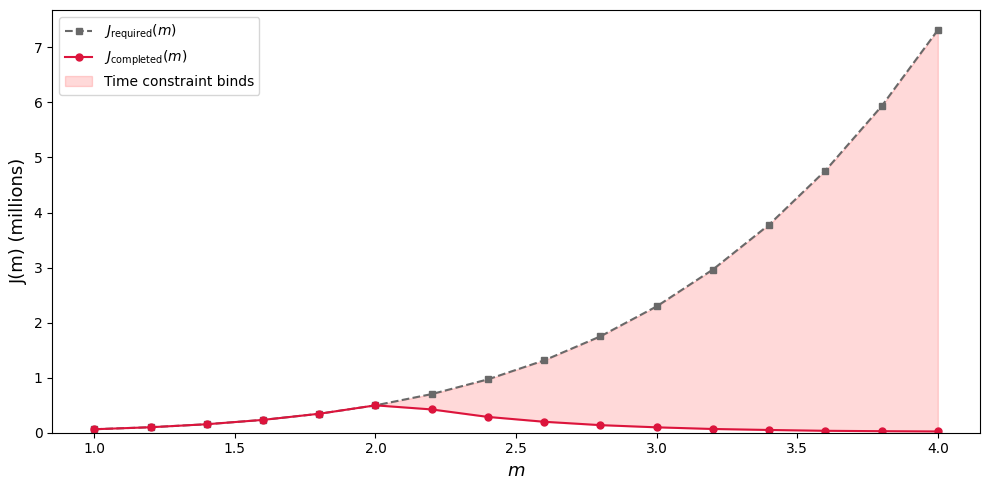}
\caption{Required iterations vs. completed within five-minute limit for $m \in \{1, \dots, 4\}$.}
\label{fig: J_required_vs_completed_4}
\end{center}
\end{figure}

Figure \ref{fig: risk_5_min} displays different risk values: analytical bounds ($\bar{v}$ and $\frac{1}{1.25}\bar{v}$), and the grid-evaluated worst-case risk of the stochastic mirror ascent algorithm.

\begin{figure}[h!] 
\begin{center}	
\includegraphics[width=3.75in]{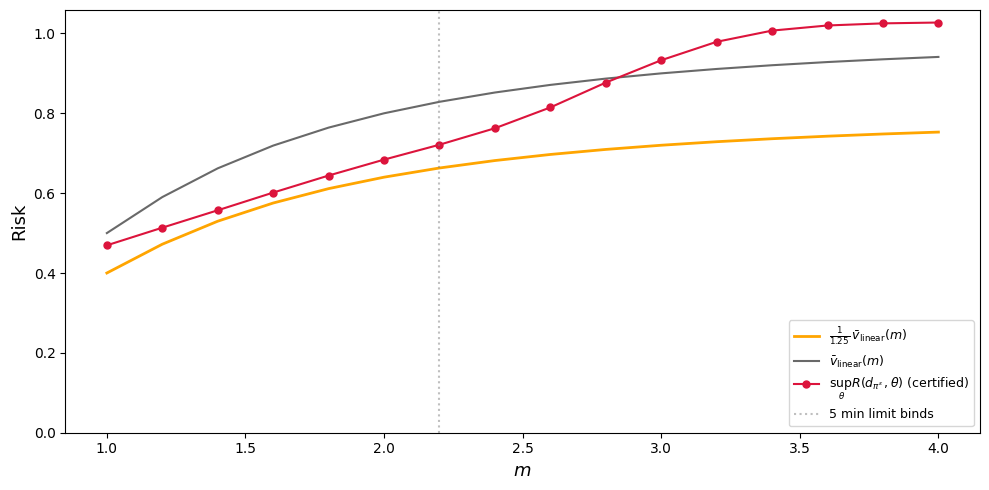}
\caption{Risk of $d_{\pi^\epsilon}$ vs. linear minimax (five-minute limit).}
\label{fig: risk_5_min}
\end{center}
\end{figure}

Figure \ref{fig: risk_impr_5_15_min} illustrates the percentage improvements in risk, analogous to Figure \ref{fig: improvement_all}, but under time restrictions: a five-minute limit (solid red line) and a fifteen-minute limit (dashed red line). Even with the five-minute limit binding for $m > 2.2$, the proposed algorithm still delivers a meaningful improvement in the worst-case risk up to $m \approx 3$. Beyond this point, the five-minute constraint is no longer sufficient to certify any improvement, and the SMA-based decision rule is dominated by the linear minimax rule. Increasing the limit to fifteen minutes pushes the break-even point out to $m \approx 3.3$ and recovers additional risk improvement throughout the upper range of $m$. However, at large $m$ the problem becomes hard enough that this time is not sufficient to beat the linear rule.

\begin{figure}[h!] 
\begin{center}	
\includegraphics[width=3.75in]{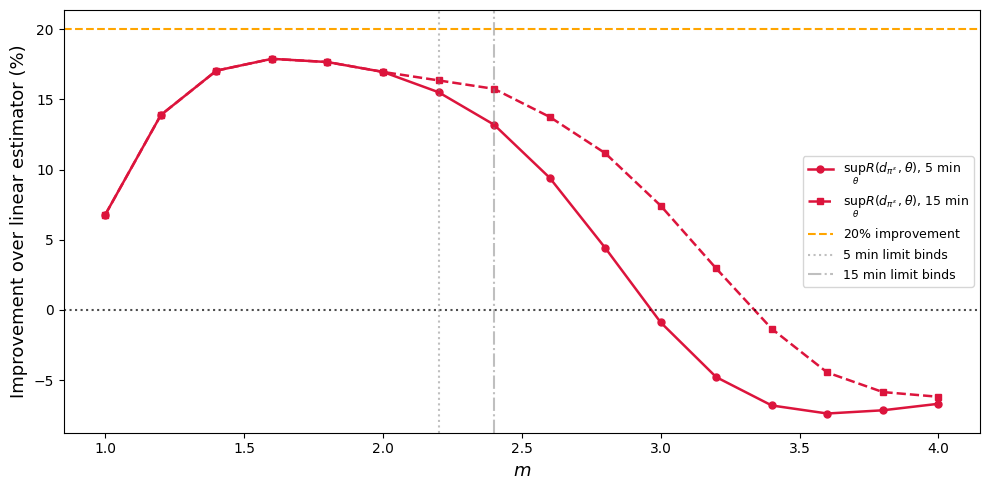}
\caption{Risk improvement of proposed estimator vs. linear (five- and fifteen-minute limits).}
\label{fig: risk_impr_5_15_min}
\end{center}
\end{figure}

\subsection{Adaptive algorithms}
We conclude this section by making a brief comment about \emph{adaptive} stochastic gradient methods---see \cite{leijordan2020} for a recent discussion. Our interest is centered on discussing potential applications of adaptive algorithms to the Bounded Normal Mean problem. As mentioned earlier, the algorithm we developed in the paper only exploits the fact that the objective function defining the least-favorable distribution is concave and Lipschitz. More concretely, our algorithm  used a simple characterization of the supergradient and a bound on the Lipschitz constant ($M:=4m^2$). However, it is entirely possible that the objective function has additional properties that could be exploited by an iterative optimization routine; for instance, the function could be strongly concave, or it could possess well-behaved higher-order derivatives, or it could have a tighter modulus of Lipschitz continuity. 

For strong concavity, it would be interesting to study the extent to which the algorithms in \cite{leijordan2020}---which adapt to an unknown modulus of strong convexity---could be used in our problem.

For adaptation to the problem's Lipschitz constant, one could explore the possibility of using the algorithms in \cite{duchi2011}, in which the step sizes are computed from the observed gradients. Interestingly, these algorithms could yield theoretical guarantees that are comparable to the ones obtained with the optimal parameters chosen a priori. This type of adaptation seems particularly promising in our problem, as fewer iterations could suffice to achieve the same accuracy. In Appendix \ref{sec: appendix_d} we present some numerical evidence that the improvement is non-negligible. To show this, we run Algorithm \ref{alg:SMA} for $m \in \{1, 1.6, 2\}$ and record the magnitude of the stochastic supergradient at every iteration. This exercise shows that replacing the worst-case Lipschitz bound with its empirical counterpart would reduce the number of iterations by factors of approximately 2, 3.5, and 5 (see Figure \ref{fig: J_of_M} in Appendix \ref{sec: appendix_d}), with the improvement increasing in $m$. We note, however, that extending the high-probability guarantees in Theorem \ref{thm:avg_prior_pointwise_LF} to data-dependent step sizes would require a different type of concentration argument, which is beyond the scope of this paper.

\section{Application} \label{sec: LP-VAR}

\subsection{Motivation: Estimation of Impulse Response Coefficients} 

A common problem in macroeconometrics is the estimation of the dynamic causal effects---or \emph{impulse response functions}---of aggregate shocks to macroeconomic policies or economic fundamentals. The two most common empirical methods currently used by applied macroeconomists for doing so are \emph{structural vector autoregressions} (henceforth interchangeably referred to as VARs or SVARs), going back to the seminal work of \cite{sims_1980}, and the so-called \emph{local projections} (LPs), introduced by \cite{jorda_2005}. 

Understanding the properties of these two estimation procedures has been an active area of recent research; see the references and discussion in \cite{montielolea_plagborgmoller_2021}, \cite{montielolea_qian_wolf_plagborgmoller_2025}, and \cite{jorda2025local}. It is now understood that there is a sense in which VARs and LPs share the same estimand \citep{plagborgmoller_wolf_2021,xu2026local}, but that in finite samples these procedures lie on opposite ends of a bias-variance trade-off \citep{li2024local}: LP estimators have lower bias than VAR estimators, but they also have substantially higher variance at intermediate and long horizons. 

Let $\theta_h$ denote the impulse response coefficient of interest at horizon $h$. For the sake of notation, we omit any explicit reference to the horizon of interest in the estimators and the covariance matrix (although both of these clearly depend on $h$). We consider a set-up where an applied macroeconomist has access to both LP and VAR estimators of $\theta_h$. We assume that these estimators follow the simple parametric model 
\begin{equation} \label{eq:bivariate_normal_model}
    \begin{pmatrix}
        \hat{\theta}_{\textrm{LP}} \\
        \hat{\theta}_{\textrm{VAR}}
    \end{pmatrix}
    \sim \mathcal{N} \left(
    \begin{pmatrix}
        \theta_h \\
        \theta_h + b
    \end{pmatrix},
    \Sigma \right), \: \Sigma \equiv
    \begin{pmatrix}
        \sigma_{\textrm{LP}}^2 & \rho \cdot \sigma_{\textrm{LP}} \sigma_{\textrm{VAR}} \\
        \rho \cdot \sigma_{\textrm{LP}} \sigma_{\textrm{VAR}}  & \sigma_{\textrm{VAR}}^2
    \end{pmatrix}.
\end{equation}
We treat the matrix $\Sigma$ as known, and assume that the variance of the LP estimator $(\sigma^2_{\textrm{LP}})$ is strictly larger than the variance of the VAR estimator $(\sigma^2_{\textrm{VAR}})$. The unknown parameter $b$ is used to represent the potential bias of the VAR estimator. 

While the exact bivariate normality of the estimators is unlikely to hold in finite samples, the asymptotic normality of the estimators can be established by using the  \emph{locally misspecified}
vector autoregression model used in \cite{montielolea_plagborgmoller_qian_wolf_forthcoming}.\footnote{\cite{montielolea_plagborgmoller_qian_wolf_forthcoming} consider stationary data generating processes that are well
approximated by a finite-order SVAR model, but subject to local misspecification in the form
of an asymptotically vanishing moving average process of potentially infinite order.} In particular, their Proposition 3.1 provides an asymptotic expansion that shows that the large-sample distribution of the LP estimator is unaffected by local misspecification (and so, there is a sense in which LPs are asymptotically unbiased). In contrast, their Proposition 3.2 shows that local misspecification typically leads to asymptotic bias in the VAR estimator. This bias is captured by the scalar parameter $b$ in \eqref{eq:bivariate_normal_model}, which will depend on the magnitude of the local misspecification. When the local misspecification is of order $O_{P}(1/\sqrt{T})$, a martingale central limit theorem can then be applied to show that the right-hand side of \eqref{eq:bivariate_normal_model} provides a valid asymptotic approximation to the finite-sample distribution of the VAR and LP estimators under the locally misspecified VAR model. 

We make two additional remarks about the structure of the bivariate normal distribution in \eqref{eq:bivariate_normal_model}. First, when the bias parameter $b$ equals zero, the estimator $\alpha \hat{\theta}_{\textrm{LP}} + (1-\alpha) \hat{\theta}_{\textrm{VAR}}$ is also unbiased for the parameter of interest $\theta_h$ (this holds for any positive or negative $\alpha \in \mathbb{R}$). Since we want to treat the VAR estimator as the efficient estimator in the absence of potential bias, this means that the linear estimator with the smallest possible variance must coincide with the VAR estimator. This is enforced by setting $\rho = \sigma_{\textrm{VAR}}/\sigma_{\textrm{LP}}$ in \eqref{eq:bivariate_normal_model}, which implies that the covariance between the two estimators is $\sigma^2_{\textrm{VAR}}$. We note that this result can also be established more generally under the locally misspecified
vector autoregression model used in \cite{montielolea_plagborgmoller_qian_wolf_forthcoming}; see their Corollary A.2. 

Second, throughout the rest of this section, we will assume that there is a known bound on the bias parameter $b$ that enters the model \eqref{eq:bivariate_normal_model}. The motivation for this bound is the following. Let $T$ denote the sample size used to produce the VAR estimator.  \cite{montielolea_plagborgmoller_qian_wolf_forthcoming} showed that if the researcher is willing to place an upper bound on the fraction of the variance of the moving average residual that is explained by the moving average component of their locally misspecified VAR (denote this parameter by $\mathcal{M}$), then the worst-case bias (in absolute value) for the VAR estimator takes the following form: 
\begin{equation} \label{eq:worst-case-bias}
\underbrace{\sqrt{T \times \mathcal{M}}}_{:=m} \:  \sigma_{\Delta}, \quad \textrm{ where } \quad  \sigma_{\Delta}:=\sqrt{\sigma_{\textrm{LP}}^2 - \sigma_{\textrm{VAR}}^2}.
\end{equation}
See Proposition 4.1 of \cite{montielolea_plagborgmoller_qian_wolf_forthcoming} and also Equation 3.3 in \cite{montielolea_qian_wolf_plagborgmoller_2025}. 

\subsection{Minimax combinations of VAR and LP estimators} \label{subsec: lp_var_minimax_combination}
The question of interest is how to best aggregate the information in LP and VAR estimators to estimate the impulse response coefficient $\theta_h$. This question has been considered recently by \cite{nemtyrev2026targeted} and \cite{chen_pesavento_vonnak_2026}. More generally, the question of how to best aggregate the information in two estimators---one of which is potentially biased but more efficient---has been analyzed by \cite{armstrong_kline_sun_2025} and \cite{lin2026introducing}. A related problem is also analyzed in \cite{wuetal2026}.

We follow closely the work of \cite{armstrong_kline_sun_2025} and consider the problem of constructing a minimax estimator (or decision rule) based on quadratic loss, and the bivariate Gaussian model in \eqref{eq:bivariate_normal_model}. We want to impose the bound \eqref{eq:worst-case-bias} on the bias parameter $|b| \leq m \cdot \sigma_{\Delta}$. 

It is without loss of generality to focus on estimators $d_h(\hat{\theta}_{\textrm{LP}}, \Delta)$ that depend on the VAR and LP estimators only through the statistics

\begin{equation*}
    \begin{pmatrix}
        \hat{\theta}_{\textrm{LP}} \\
        \Delta
    \end{pmatrix}
    \sim \mathcal{N} \left(
    \begin{pmatrix}
        \theta_h \\
        b/\sigma_{\Delta}
    \end{pmatrix},
    \begin{pmatrix}
        \sigma_{\textrm{LP}}^2 & -\sigma_{\Delta} \\
        -\sigma_{\Delta}  & 1
    \end{pmatrix} \right), \quad \Delta \equiv \left( \hat{\theta}_{\textrm{VAR}}-\hat{\theta}_{\textrm{LP}} \right) / \sigma_{\Delta}, 
\end{equation*} 
where the absolute value of $\Delta$ can be viewed as a \cite{hausman1978specification} test of correct specification of the VAR model that compares the VAR and LP impulse response estimates. A test of this kind was proposed by \cite{stock2018identification} for testing the invertibility of a vector autoregressive model. It will be convenient to introduce an auxiliary random variable $V \sim \mathcal{N}(0,\sigma^2_{\textrm{LP}}-\sigma^2_{\Delta})$, assumed to be independent of $\Delta$. We can use this random variable to construct a simple representation of the estimation error in the LP estimator. In particular, if we define $U := V - \sigma_{\Delta}(\Delta - b/\sigma_{\Delta}) \sim \mathcal{N}(0, \sigma^2_{\textrm{LP}} )$, then $(U,\Delta)$ has the same joint distribution as $(\hat{\theta}_{\textrm{LP}}-\theta_h,\Delta)$.   

We follow \cite{armstrong_kline_sun_2025} and consider the following \emph{invariance} restriction on the estimators $d: \mathbb{R}^{2} \rightarrow \mathbb{R}$:
\begin{equation} \label{eq:invariance} 
d_h(s+s',t) = d_h(s,t) + s'  \quad \forall \quad s,s',t \in \mathbb{R}. 
\end{equation}
This restriction means that if we shift the unbiased LP estimator by $s'$ units (but we keep the Hausman statistic constant), then the estimator $d$ should also be adjusted by $s'$ units. Note that, by definition of $\Delta$, any estimator of $\theta_h$ that tries to additively ``bias-adjust'' the VAR estimator by means of a function $\hat{b}(\cdot)$ that only depends on $\Delta$ satisfies the invariance restriction in \eqref{eq:invariance}. To see this, note that 
\begin{eqnarray*}
\hat{\theta}_{\textrm{VAR}}-\hat{b}(\Delta) &=& \hat{\theta}_{\textrm{LP}}+\sigma_{\Delta} \Delta - \hat{b}(\Delta) \\
& & (\textrm{by definition of $\Delta$}) \\
&=& \hat{\theta}_{\textrm{LP}} - \left( \hat{b}(\Delta) -\sigma_{\Delta} \Delta \right).
\end{eqnarray*}
Note also that any invariant estimator $d_h(\hat{\theta}_{\textrm{LP}},\Delta)$ can be written as a ``bias-adjusted'' VAR estimator: \[d_h(\hat{\theta}_{\textrm{LP}},\Delta)=d_h(0,\Delta)+\hat{\theta}_{\textrm{LP}}=\hat{\theta}_{\textrm{VAR}}+\underbrace{(-\sigma_{\Delta}\Delta+d_h(0,\Delta))}_{:=-\hat{b}(\Delta)}.\]
The risk function of estimators satisfying the invariance restriction in \eqref{eq:invariance} admits a convenient simplification:
\begin{eqnarray}
\mathcal{R}(d_h,\theta_h,b) &:=& \mathbb{E}_{\theta,b} [ (d_h(\hat{\theta}_{\textrm{LP}},\Delta)-\theta_h)^2 ] \nonumber\\
&=& \mathbb{E}_{\theta,b} [ (d_h(0,\Delta)+\hat{\theta}_{\textrm{LP}}-\theta_h)^2 ] \nonumber \\
&=& \mathbb{E}_{\theta,b} [ (d_h(0,\Delta)+V - \sigma_{\Delta}(\Delta - b/\sigma_{\Delta}))^2 ], \quad V \bot \Delta, \quad V \sim \mathcal{N}(0, \sigma^2_{\textrm{VAR}} ) \nonumber \\
&=& \sigma^2_{\textrm{VAR}} + \sigma^2_{\Delta} \mathbb{E}_{b/\sigma_{\Delta}}\left[ \left( \left[ \Delta -\frac{d_h(0,\Delta)}{\sigma_{\Delta}} \right] - \frac{b}{\sigma_{\Delta}}\right)^2 \right]. \label{eq:risk_invariant_estimators}
\end{eqnarray}
Consequently, the worst-case risk of any invariant estimator does not depend on $\theta_h$, and can be expressed in terms of the minimax risk of the Bounded Normal Mean problem we discussed in Section \ref{sec: BNM}. In particular, let $\mathcal{D}_{I}$ be the set of all invariant estimators for $\theta_h$. Then,
\begin{equation}
\bar{v}^*_h(m;\Sigma) := \adjustlimits\inf_{d_h \in \mathcal{D}_{I}}\sup_{\theta_h, |b| \leq m \cdot \sigma_{\Delta}} \mathcal{R}(d_h,\theta_h,b) = \sigma^2_{\textrm{VAR}} + \sigma^2_{\Delta} \bar{v}^*(m).
\end{equation}
Moreover, if $d^*_{\textrm{BNM-m}}$ denotes the minimax estimator for the Bounded Normal Mean problem with bound equal to $m$, the minimax invariant estimator of the parameter $\theta_h$ becomes
\begin{eqnarray}
d^*_{\textrm{h-minimax-inv}}( \hat{\theta}_{\textrm{LP}},\Delta) &=& \hat{\theta}_{\textrm{LP}} + \sigma_{\Delta} \Delta - \sigma_{\Delta} d_{\textrm{BNM-m}}^*(\Delta) \nonumber \\
&=& \hat{\theta}_{\textrm{VAR}} - \sigma_{\Delta} d_{\textrm{BNM-m}}^*(\Delta). \label{eq:minimax_estimator_final_form}
\end{eqnarray}
That is, the minimax invariant estimator for the impulse response coefficient $\theta_h$ is simply a debiased VAR estimator, where the debiasing function is based on a minimax estimator of $b/\sigma_{\Delta}$ (imposing the bound $m$). The estimator of $b/\sigma_{\Delta}$ relies only on the marginal distribution of $\Delta$.  

By the same logic as in Section \ref{sec: main_results}, the stochastic mirror ascent algorithm applied to the Bounded Normal Mean problem with bound $m$ induces, via \eqref{eq:minimax_estimator_final_form}, a stochastic minimax sequence of invariant estimators for $\theta_h$, with worst-case risk converging to $\bar{v}^*_h(m;\Sigma)$. A formal statement is given in Appendix \ref{subsec: stochastic_minimax_sec_theta}.

\subsection{Data} 
To illustrate our results, we consider four examples of impulse response estimation from four influential macroeconomic papers. In each application, impulse response estimates are obtained using both local projection and vector autoregression methods, including the same set of observable variables and controls. These estimates are based on the replication files provided by \citet{montielolea_plagborgmoller_qian_wolf_forthcoming}.\footnote{Source: \href{https://github.com/ckwolf92/lp_var_inference}{GitHub repository with replication files for \citet{montielolea_plagborgmoller_qian_wolf_forthcoming}}. The original data are from \citet{ramey_2016}.}

\begin{enumerate}
  \item \citet{gertler_karadi_2015}: This paper studies how monetary policy shocks transmit to economic and financial variables when shocks are allowed to include forward guidance. The authors identify monetary policy shocks using high-frequency interest-rate surprises around FOMC announcements as external instruments. The shock is a one-standard-deviation monetary policy tightening, which corresponds to a roughly 6-basis-point increase in the one-year government bond rate on impact. We focus on the impulse response of first-differenced log industrial production to this shock, using monthly data from January 1990 to June 2012 (270 monthly observations, or 269 after first-differencing) and including two lags. 
  
  \item \citet{romer_romer_2010}: The authors study the macroeconomic effects of tax changes by using historical records to distinguish tax actions driven by policy motives from those responding to the state of the economy. Tax shocks are identified using the narrative classification of legislated tax changes. We focus on the impulse response of log real GDP per capita to a positive exogenous tax shock---measured as an increase in tax liabilities relative to GDP---using quarterly data from 1950:Q1 to 2007:Q2 (230 data points),  including four lags and a quadratic trend. 
  
  \item \citet{ramey_2011}: This study shows that standard VAR identification may miss the timing of government-spending news, since agents may respond when news about future military spending arrives, before spending is realized. The authors identify government spending shocks using narrative measures of military spending news. We focus on the impulse response of log real GDP per capita to a government spending shock, corresponding to an increase in the expected present discounted value of future government spending associated with military events relative to previous-quarter GDP, using quarterly data from 1947:Q1 to 2013:Q2 (266 data points), including two lags and a quadratic trend. 
  \item \citet{francis_owyang_roush_dicecio_2014}: They revisit the identification of technology shocks in structural VARs by replacing infinite-horizon long-run restrictions with a finite-horizon criterion. The authors propose the ``Max Share'' technique, which identifies the technology shock as the structural shock that maximizes the explained share of the forecast-error variance of labor productivity at a long but finite horizon. We focus on the impulse response of log labor productivity to a positive unit-variance technology shock selected by this ``Max Share'' criterion, using quarterly data from 1949:Q1 to 2009:Q2 (242 data points), including two lags and a quadratic trend. 
\end{enumerate}

Figure \ref{fig: IRF_LP_VAR} shows the estimated impulse responses along with bootstrapped standard errors for both the LP and VAR estimators.

\begin{figure}[h!] 
\begin{center}	
\includegraphics[width=5in]{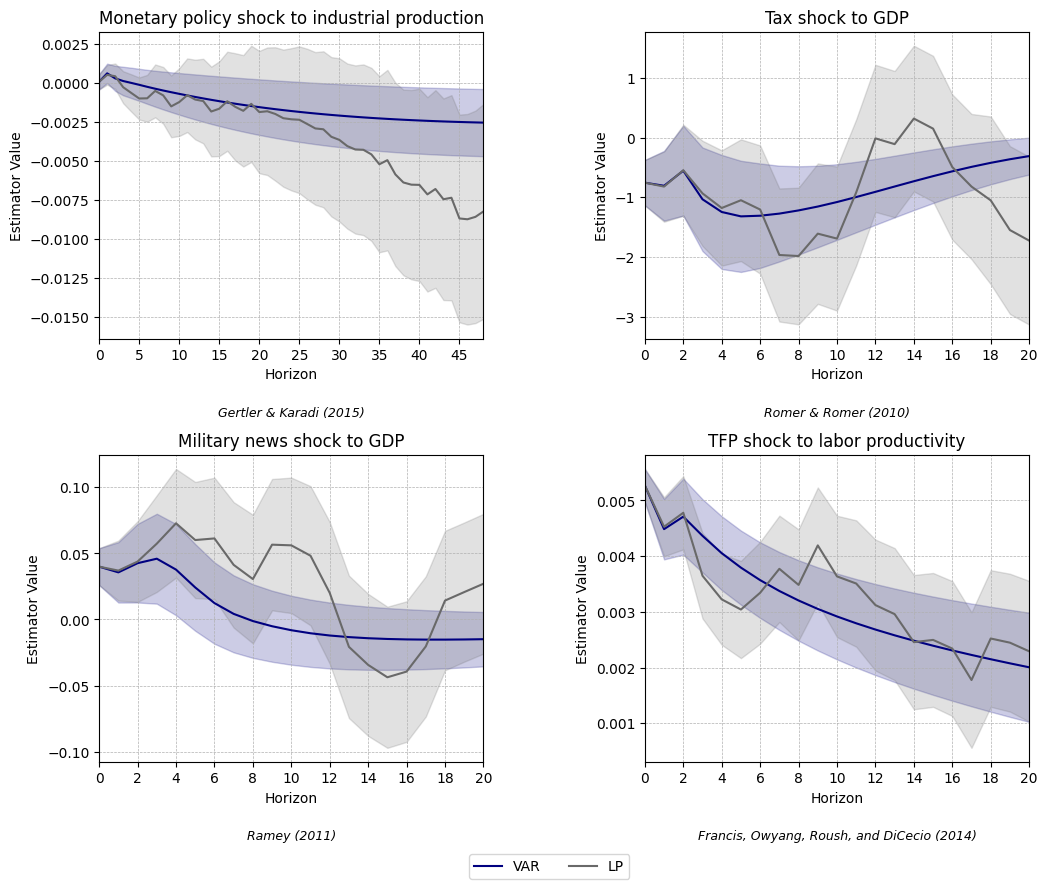}
\caption{LP vs. VAR impulse responses in four applications.}
\label{fig: IRF_LP_VAR}
\end{center}	
\noindent\footnotesize{\textit{Notes:} Each panel reports point estimates from vector autoregressions (VAR, blue) and local projections (LP, gray). Shaded areas are bootstrapped standard errors. Top-left (monthly horizons): monetary policy shock to industrial production from \citet{gertler_karadi_2015}. Other panels (quarterly horizons): tax shock to real GDP per capita from \citet{romer_romer_2010}, government spending shocks to real GDP per capita from \citet{ramey_2011}, and unanticipated TFP shocks to labor productivity from \citet{francis_owyang_roush_dicecio_2014}. }
\end{figure}

\subsection{Implementation}

We now apply the estimator in \eqref{eq:minimax_estimator_final_form} to the four applications described above. For each application and horizon $h$, we compute: (1) the \cite{hausman1978specification} specification test statistic $|\Delta_h|$, (2) the bias correction $\sigma_{\Delta_h} d_{\textrm{BNM-m}}^*(\Delta_h)$,  (3) the resulting minimax estimator as described by Equation \eqref{eq:minimax_estimator_final_form}, and (4) the percentage risk improvement relative to the minimax-optimal linear combination of LP and VAR. For reference, Table \ref{tab:m_calibration} reports the implied values of $m$ and $\epsilon(m)$ across applications. We use those values to run the stochastic mirror ascent algorithm as described in Section \ref{sec: main_results} and construct an approximate minimax decision rule corresponding to $\Delta_h$.

\begin{table}[h!]
\centering
\caption{Values of $m$ across applications ($\mathcal{M}=1$\%).}
\label{tab:m_calibration}
\begin{tabular}{lcccc}
\toprule
Application & $T$ & $m = \sqrt{T \times \mathcal{M}}$ & $\epsilon(m)$ & Time (sec)\\
\midrule
\citet{gertler_karadi_2015} & 269 & 1.64 & 0.146 & 54 \\
\citet{romer_romer_2010} & 230 & 1.52 & 0.139 & 32 \\
\citet{ramey_2011} & 266 & 1.63 & 0.145& 54 \\
\citet{francis_owyang_roush_dicecio_2014} & 242 & 1.56 & 0.142 & 40 \\
\bottomrule
\end{tabular}
\end{table}

Figure \ref{fig: Delta_stat} displays the value of the Hausman specification test statistic $|\Delta_h|$ across horizons in each application. We set the Hausman test statistic to be equal to zero at horizon $h=0$ (although it theoretically equals 0 divided by 0, since the LP and VAR estimators coincide on impact).

\begin{figure}[h!] 
\begin{center}	
\includegraphics[width=5in]{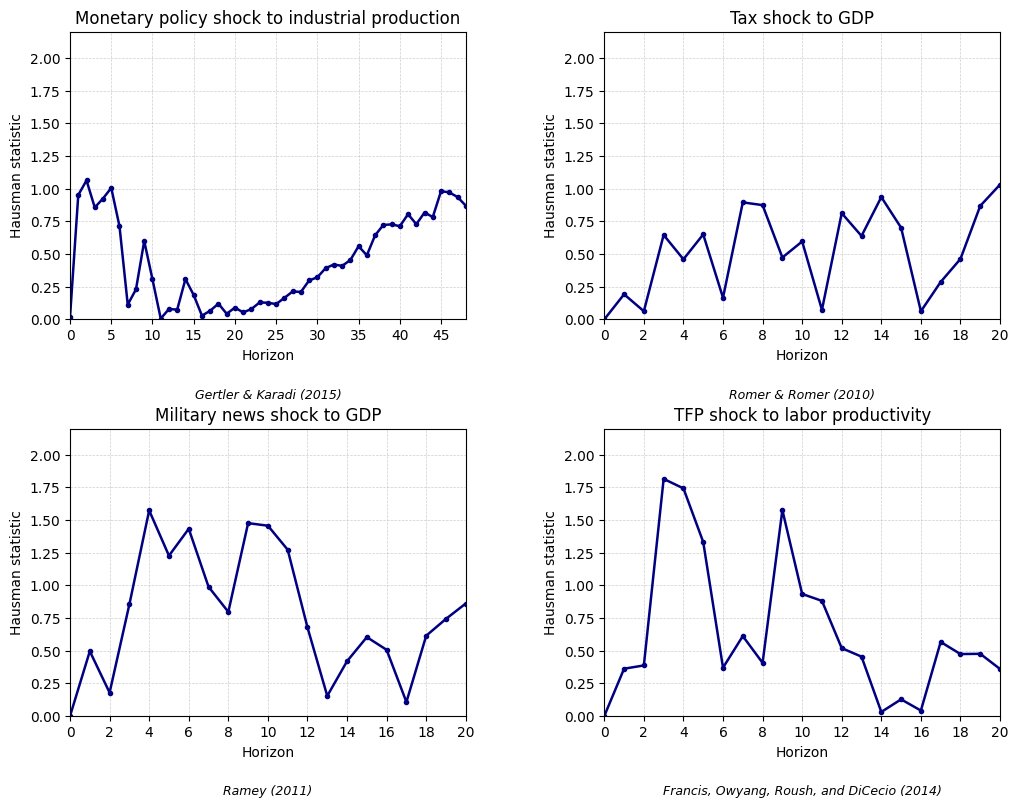}
\caption{Hausman statistic across horizons.}
\label{fig: Delta_stat}
\end{center}	
\end{figure}

For each application, we compute the approximate rule $d_{\textrm{BNM-m}}^*$ in the corresponding minimax Bounded Normal Mean problem using stochastic mirror ascent. The bias correction of the VAR estimator is nonlinear in $\Delta_h$ and is given by $\sigma_{\Delta_h} d_{\textrm{BNM-m}}^*(\Delta_h)$.

Figure \ref{fig: bias_correction} plots the implied nonlinear bias correction across horizons for the four different applications.

\begin{figure}[h!] 
\begin{center}	
\includegraphics[width=5in]{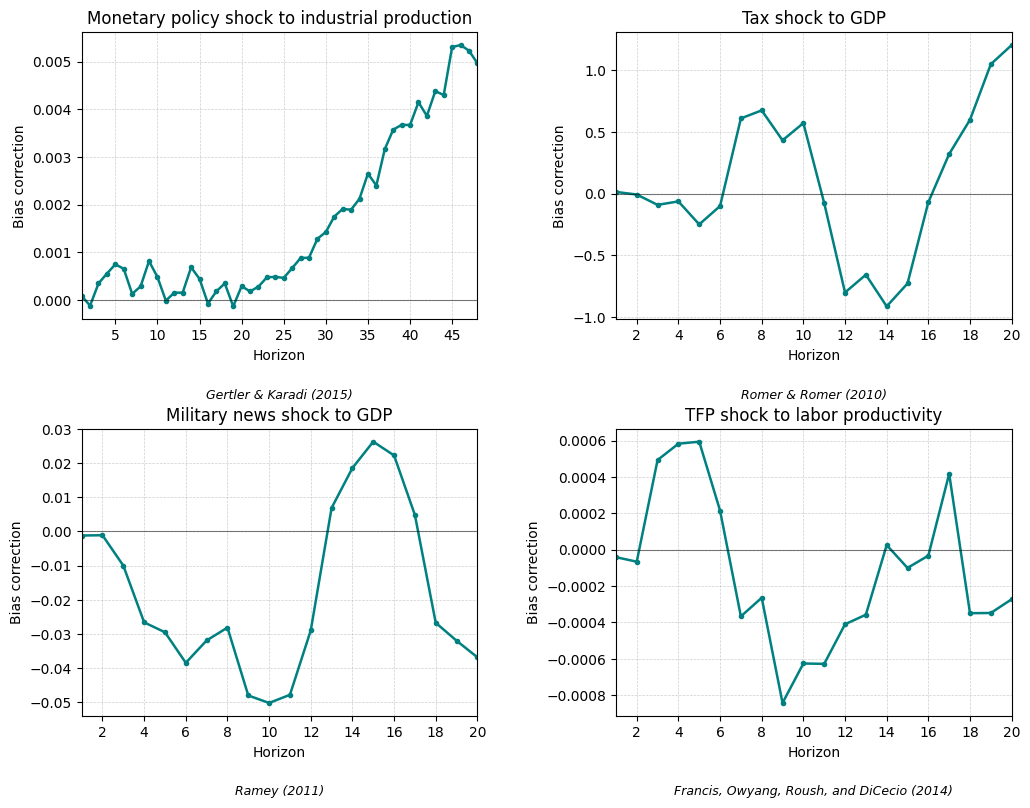}
\caption{Bias correction.}
\label{fig: bias_correction}
\end{center}	
\end{figure}

We then construct the bias-corrected VAR estimator for impulse responses in each of these applications. Figure \ref{fig: b_minimax_irfs} illustrates the comparison of this estimator to the underlying VAR and LP estimators.

\begin{figure}[h!] 
\begin{center}	
\includegraphics[width=5in]{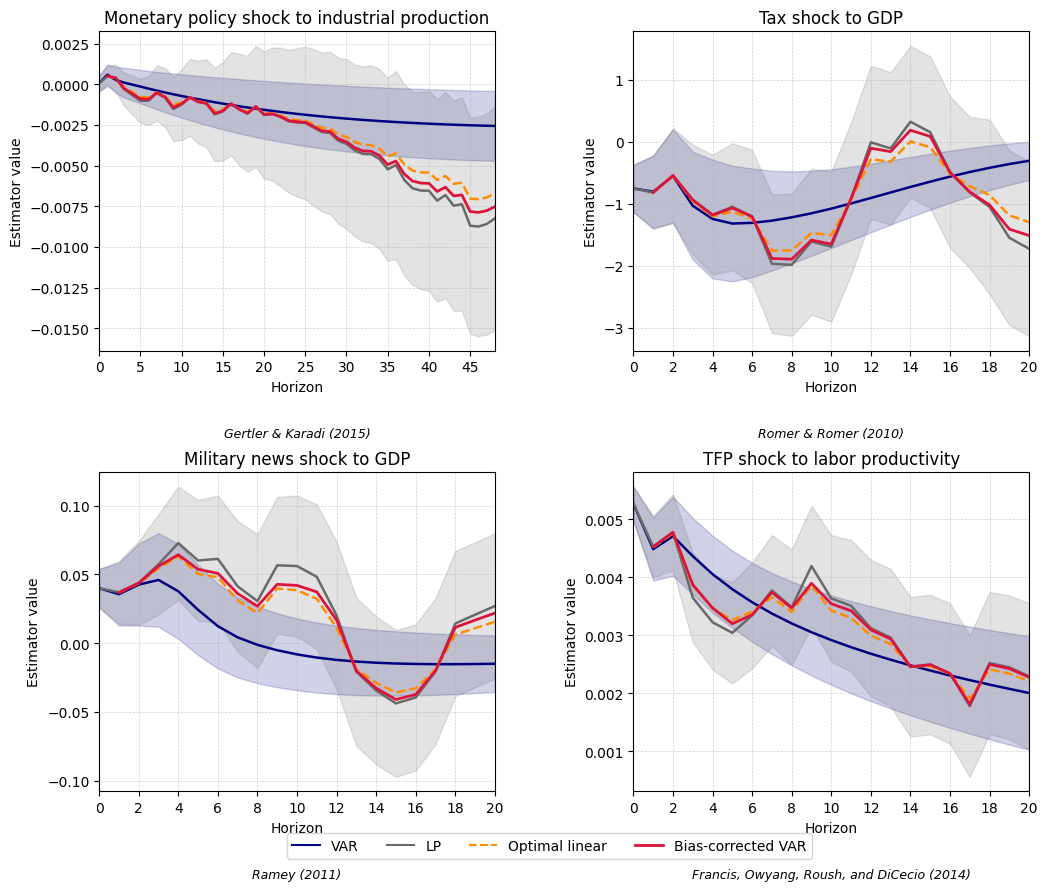}
\caption{Comparison of VAR, LP and proposed bias-corrected VAR estimators.}
\label{fig: b_minimax_irfs}
\end{center}	
\end{figure}

To benchmark the nonlinear bias-corrected VAR, we compare it with the minimax-optimal linear combination of LP and VAR. That is, a linear combination of VAR and LP estimators, with weights chosen to minimize the worst-case risk. As shown in \citet[p.~40]{montielolea_plagborgmoller_qian_wolf_forthcoming}, the worst-case risk for such a linear combination is given by
  \[(1-\omega)^2m^2(\sigma_{\textrm{LP}}^2 - \sigma_{\textrm{VAR}}^2) + \sigma_{\textrm{VAR}}^2 + \omega^2(\sigma_{\textrm{LP}}^2 - \sigma_{\textrm{VAR}}^2).\]
The optimal weight that minimizes this expression is
  \[\omega^* = \frac{m^2}{1+m^2},\]
  leading to the estimator
  \[d^*_{\textrm{linear}}=\hat{\theta}_{\textrm{VAR}} - \sigma_\Delta \omega^* \Delta.\]

In the following exercise, we compare these values with $\mathcal{R}(d_h,\theta_h,b)$, which uses the (grid-evaluated) worst-case risk of the optimal decision rule in the Bounded Normal Mean problem (1,000 Monte Carlo draws).

\begin{figure}[h!] 
\begin{center}	
\includegraphics[width=5in]{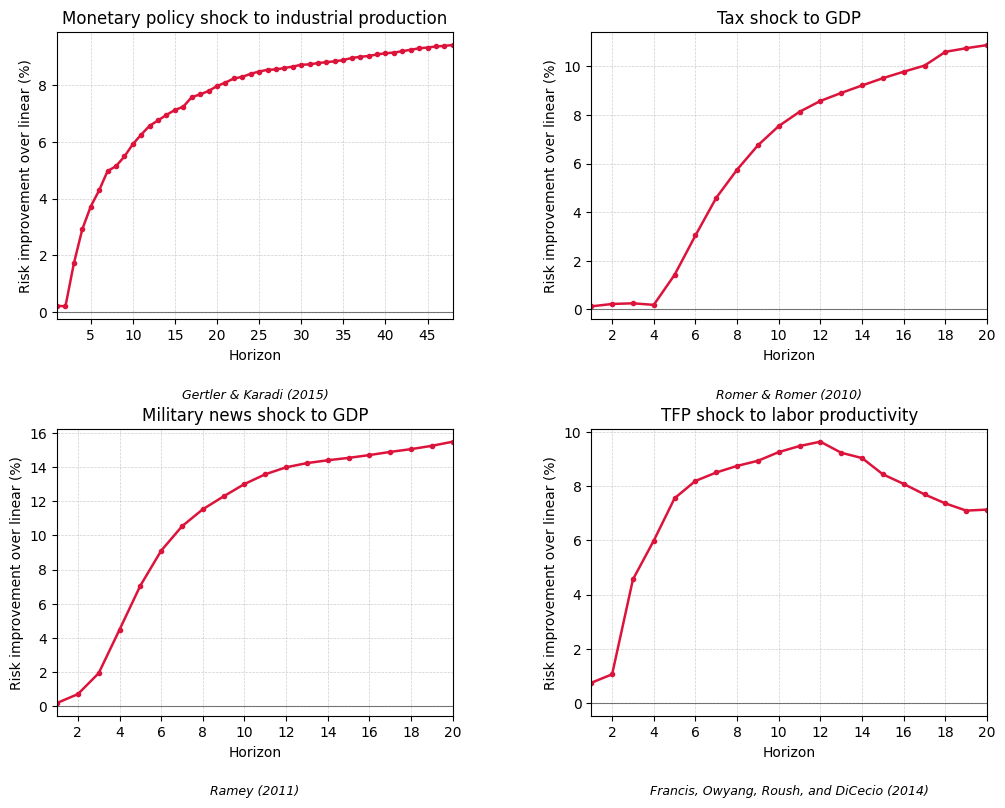}
\caption{Risk improvement over linear combination.}
\label{fig: risk_impr_irfs}
\end{center}	
\end{figure}

Figure \ref{fig: risk_impr_irfs} compares the estimated worst-case risk of the nonlinear minimax correction with the worst-case risk of this optimal linear combination. Positive values indicate that the nonlinear rule improves on the best linear rule.

\section{Conclusion}

We presented a computational approach to find an approximately minimax estimator (under quadratic loss) for the Bounded Normal Mean problem of \citet{casella_strawderman_1981}. Our suggested algorithm considers a fixed (equally-spaced) grid of $I$ points that discretize the parameter space $[-m,m]$. We then approximately maximize Bayes risk over the space of probability distributions supported on the $I$ points of the grid using a stochastic mirror ascent routine for concave maximization. We use this approximately least-favorable distribution to find the corresponding Bayes estimator. Our main results (Theorems \ref{thm:avg_prior_pointwise_LF}-\ref{thm:rule_convergence}) showed that the approximately least-favorable distribution and the corresponding Bayes estimator are indeed close (in a sense we made precise) to their desired targets.

As we explained in the paper, our algorithm has two sources of approximation error: discretization and optimization error. The Lipschitz continuity property of the risk function that we established in Lemma \ref{lemma:RiskFunction is Lips} controls the first component. The well-known theoretical properties of stochastic mirror descent allow us to control the second component. Taking into account these two sources of approximation error, we showed that the approximately least-favorable distribution obtained from our algorithm converges in $1-$Wasserstein distance to the exact least-favorable distribution of the Bounded Normal Mean problem. In addition, the approximately minimax estimator converges (uniformly over compacts) to the exact minimax estimator. 

We think that several questions remain open. First, as we have explained in the paper, our algorithm relies on an ex ante discretization of the parameter space $[-m,m]$. This makes the problem tractable and allows the use of mirror ascent on a finite-dimensional probability simplex, but it also introduces discretization error that needs to be accounted for. A natural direction for future work would be to explore the possibility of approximating the least-favorable distribution by looking directly at the concave problem of maximizing Bayes risk over all probability distributions on the interval $[-m,m]$. An approach for this problem could be based on the algorithm presented in the work of \cite{kent2021modified} for optimizing infinite-dimensional functionals of probability measures.\footnote{In Appendix \ref{section: math_properties} we derived different properties of the general maximin problem that could be helpful for this purpose.} Another possibility would be to use recent advances in online learning with continuous action spaces; see \cite{krichene2015hedge} and \cite{negrea2021minimax}. Second, some of the theoretical results we derived in order to control optimization error are based on constants/bounds that are conservative. For example, the Lipschitz constant for the risk function we derived in Lemma \ref{lemma:RiskFunction is Lips} could perhaps be tighter (which could decrease the number of iterations in Theorem \ref{thm:avg_prior_pointwise_LF}). Similarly, the number of iterations for the mirror-ascent algorithm is designed to work for any (Lipschitz) concave optimization problem over the simplex. Since we are only interested in solving the Bounded Normal Mean problem, it is perhaps possible to sharpen the analysis and reduce the number of iterations. Third, the least-favorable distribution in the Bounded Normal Mean problem is known to have a very particular structure: the support of the distribution is finite. The stochastic mirror ascent algorithm that we use is initialized at the uniform distribution in $\Delta^{I-1}$, and by design, the output of the algorithm places positive mass on every point (although as our simulations show, the mass associated with some of these points is numerically very close to zero). It could be useful to consider other algorithms for finding the least-favorable distribution that increase the number of support points in every iteration, as in the recent work of \cite{guggenberger2025numerical}.

We view our paper as part of a larger research program that emphasizes the importance of providing a careful analysis of the computational aspects surrounding  problems in econometrics and statistics; see, for example, the recent work of \citet{cattaneo2024onimplicitbias,cattaneo2025howmemory,cattaneo2025modifiedlossmomentum}. As mentioned by \citet{chandrasekaran2013computational}, \emph{``Although classical statistics gave little consideration to computational complexity, computational issues have come increasingly to the fore in modern high-dimensional statistics.''} There are many examples now within statistics and machine learning that illustrate the benefits of bringing \emph{computational awareness} to the forefront of these fields.\footnote{See, for illustration, the lectures of the \emph{``Computational Complexity of Statistical Inference Boot Camp''} of the Simons Institute for the Theory of Computing:  \href{https://www.youtube.com/playlist?list=PLgKuh-lKre13obZNGYdY-25ihfD0PIyJ-}{video lectures}.} Advances in econometrics in this regard have not followed the same pace, but as mentioned recently by Patrick Kline in one of his interventions at the Online Chamberlain Seminar reflecting on the future of econometrics,  \emph{``There is probably going to be a more important role for computation moving forward. I think we have not rewarded computational advances enough in the profession, and we ought to be more serious about that ... And it is important for procedures that are proposed to actually be implementable. And in some cases ... Having some guarantees on computability of different procedures is important and I think we ought to learn more computer science and we will.''}\footnote{See minute 19:58 of \href{https://www.youtube.com/watch?v=ifp_BUKl3ig}{this talk}} 

Finally, we hope that the methods outlined in this paper could make the application of statistical decision theory in econometrics more frequent. As noted by \citet{manski2021econometrics}, the \emph{``primary challenge to use of statistical decision theory is computational.''}

\doublespacing

\newpage

\setlength{\parindent}{20pt}
\begin{appendix}                
\begin{center}{\LARGE {\textsc{Appendix}}}\end{center}

\titleformat{\section}
  {\normalfont\Large\bfseries}
  {Appendix \thesection.}
  {1em}
  {}

\section{Proofs of Main Results} \label{sec: appendix_main_proofs}

\subsection{Proof of Lemma \ref{lemma:RiskFunction is Lips}} \label{subsec:Proof of Lemma 1} 
Let $\phi(\cdot; \theta)$ denote the p.d.f. of a univariate normal random variable with mean $\theta$ and unit variance. By definition 
\begin{eqnarray*}
R(d,\theta) &=& \int_{-\infty}^{\infty} (d(y)-\theta)^2 \phi(y;\theta)dy \\
&=& \int_{-\infty}^{\infty} d(y)^2 \phi(y;\theta)dy - 2 \theta \int_{-\infty}^{\infty} d(y) \phi(y;\theta)dy + \theta^2.
\end{eqnarray*}
Algebra then shows that for any decision rule $d \in \mathcal{D}$, and any $\theta,\theta' \in [-m,m]$ we have 
\begin{eqnarray*}
R(d,\theta) - R(d,\theta')  &=&  \int_{-\infty}^{\infty} d(y)^2 \left( \phi(y;\theta)- \phi(y;\theta') \right)dy \\
&-& 2 \theta \int_{-\infty}^{\infty} d(y) \left( \phi(y;\theta)- \phi(y;\theta') \right)dy \\
&+& 2(\theta'-\theta) \int_{-\infty}^{\infty} d(y) \phi(y;\theta')dy \\
&+& (\theta-\theta')(\theta+\theta'). 
\end{eqnarray*}
Consequently, 
\begin{eqnarray*}
\left | R(d,\theta) - R(d,\theta') \right|  & \leq &
3m^2\int_{-\infty}^{\infty} \left| \phi(y;\theta)- \phi(y;\theta') \right|dy + 4m |\theta-\theta'|. 
\end{eqnarray*}
Note that 
\[ \frac{1}{2} \int_{-\infty}^{\infty} \left| \phi(y;\theta)- \phi(y;\theta') \right|dy \]
is the total variation distance between two univariate normals, with variance 1 and means $\theta$ and $\theta'$, respectively. Therefore, Theorem 1.3 in \cite{devroye2018total} implies that 
\[ \frac{1}{2} \int_{-\infty}^{\infty} \left| \phi(y;\theta)- \phi(y;\theta') \right|dy \leq \frac{|\theta-\theta'|}{2}.  \]
Consequently, 
\[\left | R(d,\theta) - R(d,\theta') \right|   \leq 
\left(3m^2+ 4m \right) |\theta-\theta'|. \]

\subsection{Proof of Lemma \ref{lemma:discretized vs original} } \label{subsec:Proof of Lemma 2}

Note first that 
\[ \underline{v}(m;I):=  \max_{\pi\in\Delta^{I-1}} \inf_{d \in \mathcal{D}} r(d,\pi) \leq \underline{v}^*(m):=  \max_{\pi\in\Delta([-m,m])} \inf_{d \in \mathcal{D}} r(d,\pi),\] 
since $\Delta^{I-1} \subseteq \Delta([-m,m])$. Thus, the remaining part of the proof focuses on establishing the upper bound. 

For every $\theta \in [-m,m]$, let $N_{I}(\theta)$ denote a nearest neighbor of $\theta$ in the grid $\{ \theta_i  \}_{i=1}^{I}$. That is, $N_{I}(\theta)$ is any element that solves 
\[ \min_{i=1,\ldots, I} | \theta_i - \theta |. \]
Because the grid is equally-spaced with step $\mathbf{\Delta}(m,I)=2m/(I-1)$, the maximum distance from any $\theta \in [-m,m]$ to its nearest neighbor is given by
\[\sup_{\theta\in\Theta}\min_{i=1,\ldots, I} | \theta_i - \theta | =\frac{1}{2} \mathbf{\Delta}(m,I). \]
Let $d_{I}^*$ be a minimax estimator for the discretized problem, in the sense that
\[ \sup_{i=1,\ldots, I} R(d^*_{I},\theta_i) = \bar{v}(m,I):= \inf_{d \in \mathcal{D}} \sup_{i=1,\ldots, I} R(d,\theta_i). \]
By the Lipschitz continuity of the risk function established in Lemma \ref{lemma:RiskFunction is Lips} we have that for any $\theta \in [-m,m]$
\[ R(d_{I}^*,\theta) \leq R(d_I^*,N_{I}(\theta)) + (3m^2+4m) \frac{1}{2} \mathbf{\Delta}(m,I),  \]
where $\mathbf{\Delta}(m,I)$ is the gap between consecutive elements of the grid defined in \eqref{eqn: gap Delta}. Then, 
\begin{eqnarray*}
\sup_{\theta \in [-m,m]} R(d_I^*,\theta) &\leq& \sup_{\theta \in [-m,m]} R(d_I^*,N_{I}(\theta)) + (3m^2+4m) \frac{1}{2}\mathbf{\Delta}(m,I), \\
&=& \sup_{i=1,\ldots, I} R(d_I^*,\theta_i) + (3m^2+4m) \frac{1}{2}\mathbf{\Delta}(m,I)\\
&=& \bar{v}(m,I) + (3m^2+4m) \frac{1}{2} \mathbf{\Delta}(m,I).
\end{eqnarray*}
Let $d^*$ be the minimax estimator in the original Bounded Normal Mean problem. This means that
\begin{eqnarray*}
\sup_{\theta \in [-m,m]} R(d^*,\theta)  = \bar{v}^*(m) \leq \sup_{\theta \in [-m,m]} R(d_I^*,\theta)
&\leq& \bar{v}(m,I) + (3m^2+4m) \frac{1}{2} \mathbf{\Delta}(m,I).
\end{eqnarray*}
The minimax theorem holds for the Bounded Normal Mean problem $\bar{v}^*(m)= \underline{v}^*(m)$; see Proposition 4.19 in \cite{johnstone_2019} and also their Theorem A.5. And note also that the minimax theorem holds for the discretized problem (see Theorem 2.4.2 in \cite{blackwell1954} and also the discussion regarding S-games in Appendix C of \cite{montielolea_aradillasfernandez_blanchet_qiu_stoye_tan_2024}). Since the minimax theorem for the discretized problem holds, we have
\[  \underline{v}(m;I) = \inf_{d \in \mathcal{D}} \sup_{i=1,\ldots, I} R(d,\theta_i).\]
The result then follows. 

\subsection{Proof of Theorem \ref{thm:avg_prior_pointwise_LF}} \label{subsec:Proof of Theorem 1} 

Let $\pi^{*, I}$ be any maximizer of the discretized maximin problem, so that
\[ \underline{v}(m;I)=\inf_{d\in\mathcal{D}} r(d,\pi^{*, I}).\]
In a slight abuse of notation, for each epoch $j=1,\dots,J$, let $g_j$ denote the supergradient vector $g(\pi_j)$. Let $\tilde{g}_j$ be the estimator of the supergradient defined in Algorithm \ref{alg:SMA}. By definition of supergradient, we have that for every $\pi\in\Delta^{I-1}$:
\[\inf_{d\in\mathcal{D}} r(d,\pi_j) \ge \inf_{d\in\mathcal{D}} r(d,\pi) + g_j^\top(\pi_j-\pi).\]
Taking $\pi=\pi^{*, I}$ and rearranging yields
\[\underline{v}(m;I)-\inf_{d\in\mathcal{D}} r(d,\pi_j) \le g_j^\top(\pi^{*, I}-\pi_j).\]
Averaging over $j=1,\dots,J$ gives
\begin{equation}\label{eq:avg_gap_pf}
\underline{v}(m;I)-\frac{1}{J}\sum_{j=1}^J \inf_{d\in\mathcal{D}} r(d,\pi_j) \le \frac{1}{J}\sum_{j=1}^J g_j^\top(\pi^{*, I}-\pi_j).
\end{equation}
Decompose the right-hand side as
\begin{eqnarray*}
g_j^\top(\pi^{*, I}-\pi_j)
&=& \tilde{g}_j^\top(\pi^{*, I}-\pi_j) + (g_j-\tilde{g}_j)^\top(\pi^{*, I}-\pi_j) \\
&=& \tilde{g}_j^\top(\pi^{*, I}-\pi_j) + (\tilde{g}_j-g_j)^\top(\pi_j-\pi^{*, I}).
\end{eqnarray*}
Substituting into \eqref{eq:avg_gap_pf} yields
\begin{equation}\label{eq:two_terms_pf}
\underline{v}(m;I)-\frac{1}{J}\sum_{j=1}^J \inf_{d\in\mathcal{D}} r(d,\pi_j) \le
\frac{1}{J}\sum_{j=1}^J \tilde{g}_j^\top(\pi^{*, I}-\pi_j) + \frac{1}{J}\sum_{j=1}^J (\tilde{g}_j-g_j)^\top(\pi_j-\pi^{*, I}).
\end{equation}

By inequality~(4.10) in \citet{bubeck_2015} applied to the optimization domain $\Delta^{I-1}$,
\[\sum_{j=1}^J (-\tilde{g}_j)^\top(\pi_j-\pi) \le \frac{R^2}{\eta} + \frac{\eta}{2}\sum_{j=1}^J \|-\tilde{g}_j\|_\infty^2 \qquad \forall \pi\in\Delta^{I-1},\]
where $\eta$ is the step size used in the stochastic mirror descent routine and $R^2:=\ln(I)$. Taking $\pi=\pi^{*, I}$, using $\|-\tilde{g}_j\|_\infty=\|\tilde{g}_j\|_\infty\le M$ almost surely, and dividing by $J$ yields 
\begin{equation}\label{eq:regret_pf} \frac{1}{J}\sum_{j=1}^J \tilde{g}_j^\top(\pi^{*, I}-\pi_j) \le \frac{\ln(I)}{\eta J}+\frac{\eta M^2}{2}. 
\end{equation}

Next, Lemma \ref{lemma: prob_bound} in Appendix \ref{sec: aux_lemmas} applied with $\pi=\pi^{*, I}$ implies that for any $\delta>0$,
\begin{equation}\label{eq:stoch_pf}
\mathbb{P}\left( \frac{1}{J}\sum_{j=1}^J (\tilde{g}_j-g_j)^\top(\pi_j-\pi^{*, I})
<\delta \right) \ge 1-\exp\left(-\frac{\delta^2 J}{128m^4}\right).
\end{equation}
Thus, combining \eqref{eq:two_terms_pf} and \eqref{eq:regret_pf} gives
\[\frac{1}{J}\sum_{j=1}^J \inf_{d\in\mathcal{D}} r(d,\pi_j)
 \ge \underline{v}(m;I)-\frac{\ln(I)}{\eta J}-\frac{\eta M^2}{2}-\delta,\]
with probability at least $1-\exp\left(-\frac{\delta^2 J}{128m^4}\right)$. 
Finally, by concavity of Bayes risk (Lemma \ref{lemma:OptimizedBayesRule is concave}) and Jensen's inequality,
\[\inf_{d\in\mathcal{D}} r\left(d,\frac{1}{J}\sum_{j=1}^J \pi_j\right) \ge  \frac{1}{J}\sum_{j=1}^J \inf_{d\in\mathcal{D}} r(d,\pi_j).\]

Since the step size $\eta$ was chosen as $\eta(\epsilon)=(\epsilon/2)/M^2$ and the number of iterations, $J$, was set to $J(\epsilon)=2M^2\ln(I)/(\epsilon/2)^2$
\[\frac{\ln(I)}{\eta J}+\frac{\eta M^2}{2}=\frac{\epsilon}{4}+\frac{\epsilon}{4}=\frac{\epsilon}{2},\]
which yields
\[\inf_{d\in\mathcal{D}} r(d,\pi^\epsilon) \ge  \underline{v}(m;I)-\epsilon/2-\delta,\]
with probability at least $1-\exp\left(-\frac{\delta^2 J(\epsilon)}{128m^4}\right)$.
By Lemma \ref{lemma:discretized vs original} and the choice of $I(\epsilon)$ in the Theorem,
\[\underline{v}(m;I(\epsilon)) \ge \underline{v}^*(m)-\epsilon/2,\]
so the previous expression becomes
\[\inf_{d\in\mathcal{D}} r(d,\pi^\epsilon) \ge \underline{v}^*(m)-\epsilon - \delta.\]
To finish the proof, pick $\alpha \in (0,1)$ and set 
\[ \delta(\alpha,\epsilon) := \sqrt{ \frac{128m^4 \ln(1/\alpha)}{J(\epsilon)} }. \]
Note that 
\[ \exp\left(-\frac{\delta(\alpha,\epsilon)^2 J(\epsilon)}{128m^4}\right) = \alpha. \]
Algebra shows that 
\begin{equation*}
\delta(\alpha,\epsilon) \leq \frac{\epsilon}{2} \sqrt{\frac{128m^4 \ln(1/\alpha)}{2M^2 \ln(I(\epsilon))}} =  \frac{8m^2}{M} \frac{\epsilon}{2} \sqrt{\frac{\ln(1/\alpha)} {\ln(I(\epsilon))}} = \epsilon \sqrt{\frac{\ln(1/\alpha)} {\ln(I(\epsilon))}}, 
\end{equation*}
where the last equality uses the definition of $M:=4m^2$. We conclude that, with probability at least $1-\alpha$, 
\[\inf_{d\in\mathcal{D}} r(d,\pi^\epsilon) \ge \underline{v}^*(m)-\epsilon \left(1 + \sqrt{\frac{\ln(1/\alpha)} {\ln(I(\epsilon))}} \right).\]

\subsection{Proof of Theorem \ref{thm:convergence of pi_n}} \label{subsec:Proof of Theorem 2} 

Lemma \ref{lemma:OptimizedBayesRisk is Lips} shows that the map
\[f: \Delta([-m,m]) \to \mathbb{R}_+, \quad
f(\pi):=\inf_{d\in\mathcal{D}} r(d,\pi),\]
is Lipschitz continuous with respect to the 1-Wasserstein distance $W(\cdot,\cdot)$. Fix $\xi>0$ and define the set
\[S(\xi):=\{\pi\in\Delta([-m,m]): W(\pi,\pi^*) > \xi\}.\]
Since $\pi^*$ is the unique least-favorable distribution (Proposition 4.19 in \cite{johnstone_2019}), it is the unique maximizer of $f(\cdot)$.  Therefore
\[\sup_{\pi\in S(\xi)} f(\pi) < f(\pi^*)=\underline{v}^*(m).\]
The inequality above implies the existence of a scalar $\nu(\xi)>0$ such that
\[\sup_{\pi\in S(\xi)} f(\pi) < \underline{v}^*(m)-\nu(\xi).\]
Therefore,  for any $\xi>0$, 
\begin{eqnarray*}
P_n \left( W(\pi_n,\pi^*)>\xi \right) &\leq& P_n \left( f(\pi_n) < \underline{v}^*(m)-\nu(\xi)   \right) \\
&=& 1-  P_n \left( f(\pi_n) \geq \underline{v}^*(m)-\nu(\xi)   \right) \\
&=& 1-  P_n \left( \underline{v}^*(m) \geq \inf_{d\in\mathcal{D}} r(d,\pi_n) \geq \underline{v}^*(m)-\nu(\xi)   \right).
\end{eqnarray*}
It follows from Theorem \ref{thm:avg_prior_pointwise_LF} and the definition of $\pi_n$ that for any $\alpha \in (0,1)$
\[ P_n \left( \underline{v}^*(m) \geq  \inf_{d\in\mathcal{D}} r(d,\pi_n) \ge \underline{v}^*(m)-\epsilon_n\left( 1 + \sqrt{ \frac{\ln(1/\alpha)}{\ln(I(\epsilon_n))} } \right) \right) \geq 1-\alpha. \] 
Since $\epsilon_n \rightarrow 0$ and $I(\epsilon_n) \rightarrow \infty$, we have that there exists $n(\xi,\alpha)$ large enough such that for $n \geq n(\xi,\alpha)$:
\[ \epsilon_n\left( 1 + \sqrt{ \frac{\ln(1/\alpha)}{\ln(I(\epsilon_n))} } \right) \leq \nu(\xi).  \]
Therefore, for $n \geq n(\xi,\alpha)$ 
\[P_n \left( \underline{v}^*(m) \geq \inf_{d\in\mathcal{D}} r(d,\pi_n) \geq \underline{v}^*(m)-\nu(\xi)   \right) \geq 1-\alpha,\]
and consequently, 
\[P_n \left( W(\pi_n,\pi^*)>\xi \right) \leq \alpha.\]

\subsection{Proof of Theorem \ref{thm:minimax sequence} } \label{subsec:Proof of Theorem 3}  

Let $\pi^*$ be the least-favorable distribution of the Bounded Normal Mean problem. Algebra shows that, for any $\theta\in[-m,m]$, 
\begin{align*}
|R(d_n,\theta)-R(d_{\pi^*},\theta)|
&= \left | \int_{-\infty}^{\infty}
\Big((d_n(y)-\theta)^2-(d_{\pi^*}(y)-\theta)^2\Big)\phi(y;\theta) dy \right | \\
&\le \int_{-\infty}^{\infty}
\left|(d_n(y)-\theta)^2-(d_{\pi^*}(y)-\theta)^2 \right|\phi(y;\theta) dy \\
&= \int_{-\infty}^{\infty}
\left|d_n(y)^2-2d_n(y) \theta-d_{\pi^*}(y)^2+2d_{\pi^*}(y) \theta \right|\phi(y;\theta) dy \\
&= \int_{-\infty}^{\infty}
|d_n(y)-d_{\pi^*}(y)| \cdot \underbrace{|d_n(y)+d_{\pi^*}(y) - 2 \theta|}_{\le 4 m} \phi(y;\theta) dy \\
&\le 4m \int_{-\infty}^{\infty}|d_n(y)-d_{\pi^*}(y)|\phi(y;\theta) dy.
\end{align*}
Theorem \ref{thm:convergence of pi_n} showed that $W(\pi_n,\pi^*) \overset{p}{\rightarrow}0$. Using Lemma \ref{lemma: uniform risk bound} and the fact that $d_n:=d_{\pi_n}$, we have that  
\[ \sup_{\theta \in [-m,m]} \int_{-\infty}^{\infty}|d_n(y)-d_{\pi^*}(y)|\phi(y;\theta) dy \overset{p}{\rightarrow} 0,   \]
which implies that 
\begin{equation} \label{eq:auxiliary Theorem 3}
\sup_{\theta \in [-m,m]}|R(d_n,\theta)-R(d_{\pi^*},\theta)| \overset{p}{\rightarrow} 0. 
\end{equation}
To finalize the proof note that
\begin{eqnarray*}
\sup_{\theta \in [-m,m]} R(d_n,\theta) &=& \sup_{\theta \in [-m,m]} R(d_n,\theta)-R(d_{\pi^*},\theta) + R(d_{\pi^*},\theta) \\
&\leq &\sup_{\theta \in [-m,m]} \left | R(d_n,\theta)-R(d_{\pi^*},\theta) \right | + \sup_{\theta \in [-m,m]} R(d_{\pi^*},\theta)\\
&=& \sup_{\theta \in [-m,m]} \left | R(d_n,\theta)-R(d_{\pi^*},\theta) \right | + \bar{v}^*(m),
\end{eqnarray*}
where the last line follows from the fact that $d_{\pi^*}$ is minimax, by Proposition 4.19 in \cite{johnstone_2019}. Therefore, for any $\varepsilon>0$
\[P_n \left( \sup_{\theta \in [-m,m]} R(d_n,\theta) \leq \bar{v}^*(m) + \varepsilon \right) \geq P_n \left( \sup_{\theta \in [-m,m]} \left | R(d_n,\theta)-R(d_{\pi^*},\theta) \right | \leq \varepsilon \right).   \]
Equation \eqref{eq:auxiliary Theorem 3} implies that $\{d_n\}_{n=1}^{\infty}$ is a stochastic minimax sequence.

\subsection{Proof of Theorem \ref{thm:rule_convergence} } \label{subsec:Proof of Theorem 4} 

Take any constant $c > 0$ such that $S \subseteq [-c,c]$. By Lemma \ref{lemma:sequential continuity}, for any $y \in \mathbb{R}$
\[|d_n(y) - d_{\pi^*}(y)| \le \frac{C(y)}{b(y,\pi_n)} W(\pi_n,\pi^*),\]

where $C(y):=C_1(y)+mC_2(y)$ and $b(y,\pi_n):=  \int_{[-m,m]} \phi(y; \theta) d \pi_n$. Since $C(\cdot)$ is continuous and $S$ is compact, $\sup_{y \in S} C(y) =: A < \infty$. 

Using the triangle inequality, for any $y \in [-c,c]$ and for any $\theta \in [-m,m]$
\[ \phi(y; \theta) = \frac{1}{\sqrt{2\pi}} \exp \left( -\frac{(y-\theta)^2}{2} \right) \ge \frac{1}{\sqrt{2\pi}} \exp \left( -\frac{(c+m)^2}{2} \right) =: B >0.\]
Since this holds for every $\theta \in [-m,m]$, integrating over any $\pi \in \Delta([-m,m])$ yields
\[b(y,\pi) =  \int_{[-m,m]} \phi(y; \theta)\, d\pi(\theta) \ge B,\]
uniformly over all $y \in S$. In particular, $b(y,\pi_n) \ge B$ for all $y \in S$.

Combining,
\[|d_n(y) - d_{\pi^*}(y)| \le \frac{A}{B} W(\pi_n,\pi^*).\]

Therefore, by Theorem \ref{thm:convergence of pi_n}, for any $\xi>0$
\[P_n \left(  \sup_{y \in S} |d_n(y) - d_{\pi^*}(y)| > \xi \right) \le P_n \left( W(\pi_n,\pi^*) > \frac{B}{A} \xi \right) \to 0. \]

\section{Main Auxiliary Lemmas} \label{sec: aux_lemmas}

\subsection{Proof of Lemma \ref{lemma: prob_bound}}

The proof of Theorem \ref{thm:avg_prior_pointwise_LF} uses the following auxiliary lemma. 

\begin{lemma}[Probability bound]\label{lemma: prob_bound} Let $\tilde{g}_j$ and $g_j$ be defined as in the proof of Theorem \ref{thm:avg_prior_pointwise_LF}. For any $\delta > 0$ and any $\pi \in \Delta^{I-1}$,
\[\mathbb{P}\left(
  \frac{1}{J}\sum_{j=1}^{J}(\tilde{g}_j - g_j)^\top(\pi_j - \pi) < \delta
\right) \ge 1 - \exp \left(-\frac{J\delta^2}{128m^4}\right).\]
\end{lemma}

\begin{proof}
For each $j = 1,\dots,J$, define $x_j := \pi_j - \pi$ and
\[ \Delta_j := \tilde{g}_j - g_j = \bigl(
(d_{\pi_j}(Y_{i,j}) - \theta_i)^2 - \mathbb{E}_{\theta_i}\bigl[(d_{\pi_j}(Y_{i,j}) - \theta_i)^2\bigr] \bigr)_{i=1}^{I},
\]
where $Y_{i,j} \sim \mathcal{N}(\theta_i,1)$ (independently of $\pi_j$. Since $d_{\pi_j}(Y_{i,j})\in[-m,m]$ and $\theta_i\in[-m,m]$, we have $(d_{\pi_j}(Y_{i,j}) - \theta_i)^2\in[0,4m^2]$ for every $i \in \{1,\ldots, I\}$. Thus, each component of $\Delta_j=\tilde{g}_j-g_j$ lies in $[-4m^2,4m^2]$, and so $\|\Delta_j\|_\infty\le 4m^2$.
Let $L:= 8 m^2$ and let $\{\mathcal{F}_j\}_{j \ge 0}$ be the natural filtration, generated by $\{(Y_{1,j}, \ldots, Y_{I,j})\}_{j \geq 0}$). Note that under the natural filtration, $\mathbb E[\Delta_j\mid\mathcal{F}_{j-1}]=0$. Moreover, since $\|x_j\|_1 \le 2$ and $\|\Delta_j\|_\infty \le 4m^2$, we have
\[
|\langle \Delta_j, x_j\rangle|
\le \|x_j\|_1 \|\Delta_j\|_\infty
\le 8m^2.
\]
Hence, conditional on $\mathcal{F}_{j-1}$, the random variable $\langle \Delta_j, x_j\rangle$
is mean zero and takes values on a subset of the interval $[-8m^2,8m^2]$.
Hoeffding's lemma (Lemma 2.2 in \citet{boucheron2013}) therefore yields, for every $\gamma > 0$,
\begin{equation}\label{eq:mgf_step}
\mathbb{E}\bigl[ e^{\gamma\langle\Delta_j,x_j\rangle} \mid \mathcal{F}_{j-1} \bigr] \le \exp \left(\frac{\gamma^2(16m^2)^2}{8}\right) = \exp \left(\frac{\gamma^2 L^2}{2}\right).
\end{equation}

We now bound the moment generating function of the average $\frac{1}{J}\sum_{j=1}^{J} \langle \Delta_j, x_j \rangle$ by the law of iterated expectations. Conditioning on $\mathcal{F}_{J-1}$ and applying~\eqref{eq:mgf_step},
\begin{align*}
\mathbb{E} \left[ \exp \left(\frac{\gamma}{J}\sum_{j=1}^{J}\langle\Delta_j, x_j\rangle\right)
\right]
&= \mathbb{E}\left[ \exp \left(\frac{\gamma}{J}\sum_{j=1}^{J-1}\langle\Delta_j, x_j\rangle\right) \cdot
\mathbb{E}\left[ \exp\left(\frac{\gamma}{J}\langle\Delta_J, x_J\rangle\right) \mid \mathcal{F}_{J-1} \right] \right] \\
&\le \exp \left(\frac{\gamma^2 L^2}{2J^2}\right) \cdot
\mathbb{E}\left[ \exp \left(\frac{\gamma}{J}\sum_{j=1}^{J-1}\langle\Delta_j, x_j\rangle\right) \right].
\end{align*}

Repeating the argument sequentially from $j = J$ to $j = 1$ gives
\[ \mathbb{E} \left[ \exp\left(\frac{\gamma}{J}\sum_{j=1}^{J}\langle\Delta_j, x_j\rangle\right)
\right] \le \prod_{j=1}^{J}\exp\left(\frac{\gamma^2 L^2}{2J^2}\right) = \exp\left(\frac{\gamma^2 L^2}{2J}\right).\]

By Markov's inequality, for any $\gamma > 0$ and $\delta > 0$,
\begin{align*}
\mathbb{P} \left( \frac{1}{J} \sum_{j=1}^{J} \langle\Delta_j, x_j\rangle \ge \delta \right) &\le \exp\left(\frac{\gamma^2 L^2}{2J} - \gamma\delta\right) \\
\implies \mathbb{P}\left( \frac{1}{J}\sum_{j=1}^{J}(\tilde{g}_j - g_j)^\top(\pi_j - \pi)  < \delta \right) & \ge 1 - \exp \left(\frac{\gamma^2 L^2}{2J} - \gamma\delta\right). 
\end{align*}

Minimizing the exponent over $\gamma > 0$ yields $\gamma^* = \delta J / L^2$, and substituting into the expression above gives the resulting probability bound
\[ \exp \left( \frac{\delta^2 J}{2L^2} - \frac{\delta^2 J}{L^2} \right) = \exp\left(-\frac{J\delta^2}{2L^2}\right)
= \exp\left(-\frac{J\delta^2}{128m^4}\right).
\] 
\end{proof}

\subsection{Proof of Lemma \ref{lemma: uniform risk bound}} \label{subsec: uniform risk bound}

The proof of Theorem \ref{thm:minimax sequence} uses the following auxiliary Lemma. 

\begin{lemma} \label{lemma: uniform risk bound}
Let $\{\pi_n\}_{n=1}^{\infty}$ be a stochastic sequence of probability distributions in $\Delta([-m,m])$ such that $W(\pi_n,\pi) \overset{p}{\rightarrow}0$. Let $d_n:=d_{\pi_n}$ denote the Bayes rule associated with $\pi_n$. Then, for any $\xi>0$
\[ P_n \left( \sup_{\theta \in [-m,m]}\int_{-\infty}^{\infty}|d_n(y)-d_{\pi}(y)|\phi(y;\theta) dy > \xi \right) \rightarrow 0 \]
as $n \rightarrow \infty$. 
\end{lemma}

\begin{proof} The proof of Lemma \ref{lemma:sequential continuity} showed that for any $y \in \mathbb{R}$ and any sequence $\{\pi_n\}_{n=1}^{\infty}$
\begin{eqnarray*}
|d_{\pi_n}(y)-d_{\pi}(y) | &\leq& \frac{1}{b(y,\pi_n)} \left( C_1(y)  +  m C_2(y) \right) W(\pi_n,\pi),
\end{eqnarray*}
where 
\[ C_1(y):=\frac{1+m(|y|+m)}{\sqrt{2 \pi}} , \quad C_2(y):=\frac{|y|+m}{\sqrt{2 \pi}}, \quad b(y,\pi_n):= \int_{[-m,m]} \phi(y;\theta)d\pi_n. \]
Let $C(y):=C_1(y)+mC_2(y)$. Then, we have 
\begin{equation} \label{eq:auxiliary bound d_{pi} diff} 
|d_{\pi_n}(y)-d_{\pi}(y) | \leq \frac{C(y)}{ b(y,\pi_n)}  W(\pi_n,\pi). 
\end{equation}
Fix $\alpha \in (0,1)$ and let $z_{1-\alpha/2}$ be the $1-(\alpha/2)$ quantile of a standard normal. For any $\theta \in [-m,m]$ we have
\begin{eqnarray}
\int_{-\infty}^{\infty}|d_n(y)-d_{\pi}(y)|\phi(y;\theta) dy &=& \int_{-\infty}^{\theta-z_{1-\alpha/2}} |d_n(y)-d_{\pi}(y)|\phi(y;\theta) dy \nonumber\\
&+& \int_{\theta-z_{1-\alpha/2}}^{\theta+z_{1-\alpha/2}} |d_n(y)-d_{\pi}(y)|\phi(y;\theta) dy \nonumber \\
&+& \int_{\theta+z_{1-\alpha/2}}^{\infty} |d_n(y)-d_{\pi}(y)|\phi(y;\theta) dy \nonumber\\
& \leq & \int_{\theta-z_{1-\alpha/2}}^{\theta+z_{1-\alpha/2}} |d_n(y)-d_{\pi}(y)|\phi(y;\theta) dy + 2m \alpha \label{eq:auxiliary bound d_{pi} diff 2}
\end{eqnarray}
where the last inequality follows from the fact that $|d_n(y)-d_\pi(y)|$ is bounded by $2m$. Note also that for any $y \in [\theta-z_{1-\alpha/2},\theta+z_{1-\alpha/2}]$ and any $\pi \in \Delta([-m,m])$ we have 
\[  b(y,\pi) \geq C_3(\alpha):= \frac{1}{\sqrt{2\pi}}\exp\left( - \frac{1}{2} (2m+z_{1-\alpha/2})^2 \right).  \]
Therefore, \eqref{eq:auxiliary bound d_{pi} diff} and \eqref{eq:auxiliary bound d_{pi} diff 2} imply
\begin{eqnarray*}
\int_{-\infty}^{\infty}|d_n(y)-d_{\pi}(y)|\phi(y;\theta) dy &\leq&  2m\alpha   +  \frac{1}{ C_3(\alpha)}  W(\pi_n,\pi) \int_{-\infty}^{\infty} C(y) \phi(y;\theta) dy.
\end{eqnarray*}
Since $\alpha$ is arbitrary, $W(\pi_n,\pi) \overset{p}{\rightarrow}0$, and 
\[ C:= \sup_{\theta \in [-m,m]} \int_{-\infty}^{\infty} C(y) \phi(y;\theta) dy < \infty, \]
we conclude that for any $\xi>0$
\[ P_n \left( \sup_{\theta \in [-m,m]}\int_{-\infty}^{\infty}|d_n(y)-d_{\pi}(y)|\phi(y;\theta) dy > \xi \right) \rightarrow 0 \]
as $n \rightarrow \infty$. 
\end{proof} 

\subsection{Proof of Lemma \ref{lem: stochastic_minimax_sec_theta}} \label{subsec: stochastic_minimax_sec_theta}

The following lemma provides a formal statement of the stochastic sequence of invariant estimators result discussed in Section \ref{subsec: lp_var_minimax_combination}.

\begin{lemma} \label{lem: stochastic_minimax_sec_theta}
    Fix $m > 0$ and horizon $h$. Let $\{d_n\}_{n=1}^\infty$ be a stochastic minimax sequence for the Bounded Normal Mean problem:
    \[Y \sim \mathcal N(\theta, 1), \quad \theta \in [-m,m]. \]
    For each $n$, define an invariant estimator as in Equation \eqref{eq:minimax_estimator_final_form}
    \[\delta_{h,n} (\hat{\theta}_{\textrm{LP}},\Delta) = \hat{\theta}_{\textrm{VAR}} - \sigma_{\Delta} d_n(\Delta).\]
    Then $\{\delta_{h,n}\}_{n=1}^\infty$ is a stochastic minimax sequence for the LP--VAR problem over the class of invariant estimators. That is, for any $\xi > 0$,
    \[P_n \left( \sup_{\theta_h, |b| \leq m \cdot \sigma_{\Delta}} \mathcal{R}(\delta_{h,n}, \theta_h, b) > \bar{v}^*_h (m; \Sigma) + \xi \right)  \to 0, \text{ as } n \to \infty,\]
    where 
    \[\bar{v}^*_h (m; \Sigma) := \sigma_{\mathrm{VAR}}^2 + \sigma_\Delta^2 \bar{v}^*(m).\]
\end{lemma}

\begin{proof}
    Let $\theta := \frac{b}{\sigma_\Delta}$. By the maintained bound $|b| \le m \cdot \sigma_\Delta$, we have $\theta\in[-m,m]$.

By definition of $\Delta$,
\[\hat{\theta}_{\textrm{VAR}} = \hat{\theta}_{\textrm{LP}}+\sigma_\Delta \Delta.\]

Therefore,
\[\delta_{h,n}(\hat{\theta}_{\textrm{LP}},\Delta) = \hat{\theta}_{\textrm{LP}} + \sigma_\Delta \Delta - \sigma_\Delta d_n(\Delta).\]
Thus, in the notation of invariant estimator,
\[ \delta_{h,n} (0,\Delta) =\sigma_\Delta \Delta-\sigma_\Delta d_n(\Delta).\]

Using Equation \eqref{eq:risk_invariant_estimators}, we obtain
\begin{align*}
\mathcal{R}(\delta_{h,n},\theta_h,b) &= \sigma_{\textrm{VAR}}^2 + \sigma_\Delta^2 \mathbb E_{b / \sigma_\Delta} \left[ \left( \left[ \Delta- \frac{\delta_{h,n} (0,\Delta)}{\sigma_\Delta} \right]  - \frac{b}{\sigma_\Delta} \right)^2 \right] \\
&= \sigma_{\textrm{VAR}}^2 + \sigma_\Delta^2 \mathbb{E}_{\theta} \left[ \left( d_n(\Delta)-\theta \right)^2 \right].
\end{align*}

The second equality follows because
\[\Delta - \frac{\delta_{h,n} (0,\Delta)}{\sigma_\Delta} = \Delta - \frac{\sigma_\Delta\Delta -\sigma_\Delta d_n(\Delta)}{\sigma_\Delta} = d_n(\Delta).\]

Since $\Delta \sim \mathcal{N}(\theta, 1)$, the last expectation is exactly the risk of $d_n$ in the Bounded Normal Mean problem:
\[ \mathbb{E}_{\theta} \left[ \left(d_n(\Delta)-\theta\right)^2 \right] = R(d_n,\theta). \]
Hence,
\[ \mathcal{R}(\delta_{h,n},\theta_h,b) = \sigma_{\textrm{VAR}}^2+\sigma_\Delta^2R(d_n,\theta).\]

Taking supremum over $\theta_h$ and $|b| \le m\sigma_\Delta$ gives
\[ \sup_{\theta_h,\ |b|\le m\sigma_\Delta} \mathcal{R}(\delta_{h,n}, \theta_h, b) = \sigma_{\textrm{VAR}}^2 + \sigma_\Delta^2 \sup_{\theta \in [-m,m]}R(d_n, \theta), \]
using the fact that the risk of invariant estimators does not depend on $\theta_h$ after the reduction.

Therefore, for any $\xi>0$,
\begin{align*}
& P_n \left( \sup_{\theta_h, |b| \le m \cdot \sigma_\Delta} \mathcal{R}(\delta_{h,n}, \theta_h, b) >
\bar{v}_h^*(m;\Sigma) + \xi \right) \\
&=P_n \left(\sigma_{\textrm{VAR}}^2 + \sigma_\Delta^2 \sup_{\theta \in [-m,m]}R(d_n, \theta) > \sigma_{\textrm{VAR}}^2 + \sigma_\Delta^2\bar{v}^*(m) + \xi \right) \\
&= P_n \left(\sup_{\theta \in [-m,m]} R(d_n,\theta) > \bar{v}^*(m)+\frac{\xi}{\sigma_\Delta^2}\right).
\end{align*}

Since $\{d_n\}_{n=1}^{\infty}$ is a stochastic minimax sequence for the Bounded
Normal Mean problem, the last probability converges to zero. Hence
\[P_n \left(\sup_{\theta_h, |b| \le m\sigma_\Delta} \mathcal{R}(\delta_{h,n}, \theta_h, b) > \bar{v}_h^*(m;\Sigma)+\xi \right) \to 0. \]
\end{proof}

\section{Mathematical Properties of the Maximin Problem} \label{section: math_properties}
In this section, we focus on the problem of finding a least-favorable distribution $\pi^*$; that is, a distribution $\pi^*$ that solves \eqref{equation:maximin}. To this purpose, it will be convenient to define the function $f: \Delta([-m,m]) \rightarrow \mathbb{R}_{+}$ given by
\begin{equation}
f(\pi) := \inf_{d \in \mathcal{D}} r(d,\pi) = \inf_{d \in \mathcal{D}} \int_{[-m,m]} R(d,\theta) d\pi .
\end{equation}
Lemma \ref{lemma:OptimizedBayesRisk is Lips} below shows that the function $f(\cdot)$ is Lipschitz with respect to the 1-Wasserstein metric over $\Delta([-m,m])$.

\begin{lemma}[The optimized Bayes risk is Lipschitz in $\pi$] \label{lemma:OptimizedBayesRisk is Lips}
Suppose that for any $\pi \in \Delta([-m,m])$ there exists a Bayes rule $d_{\pi}$ as defined in  \eqref{eqn:Bayes}. Then, for any $\pi,\pi' \in \Delta([-m,m])$
\[ | f(\pi)-f(\pi')  | \leq (3m^2 + 4m) \cdot W(\pi,\pi').   \]
\end{lemma} 
\begin{proof} 
Let $\pi,\pi'$ be two arbitrary elements of $\Delta([-m,m])$. Without loss of generality, let us assume that $f(\pi) \geq f(\pi')$. Let $L(m):=3m^2 + 4m$ be the Lipschitz constant of Lemma \ref{lemma:RiskFunction is Lips}. Then
\begin{eqnarray*}
0 \leq f(\pi) - f(\pi') &=& \int_{[-m,m]} R(d_{\pi},\theta) d\pi -  \int_{[-m,m]} R(d_{\pi'},\theta) d\pi' \\
&\leq& \int_{[-m,m]} R(d_{\pi'},\theta) d\pi -  \int_{[-m,m]} R(d_{\pi'},\theta) d\pi' \\
&& \textrm{(by definition of $d_{\pi}$)} \\
&=& L(m) \left( \int_{[-m,m]} \frac{R(d_{\pi'},\theta)}{L(m)} d\pi -  \int_{[-m,m]} \frac{R(d_{\pi'},\theta)}{L(m)} d\pi' \right) \\
&\leq& L(m) W(\pi,\pi'),
\end{eqnarray*}
where the last inequality follows from \eqref{eqn:KR-duality} and the fact that, according to Lemma \ref{lemma:RiskFunction is Lips}, 
\[ \left\| \frac{R(d_{\pi'},\theta)}{L(m)} \right\|_{\textrm{Lip}} = \frac{1}{L(m)}\sup_{\theta \neq \theta'} \frac{|R(d_{\pi'},\theta)-R(d_{\pi'},\theta')|}{|\theta-\theta'|} \leq 1. \]
\end{proof}

Lemma \ref{lemma:OptimizedBayesRule is concave} below further shows that the optimized Bayes risk is a concave function of $\pi$ and that the risk function can be viewed as a supergradient. 

\begin{lemma}[The optimized Bayes risk is Concave in $\pi$] \label{lemma:OptimizedBayesRule is concave}
Let $\pi,\pi'$ be two arbitrary elements of $\Delta([-m,m])$. For any $\lambda \in [0,1]$, 
\[f\left( \lambda \pi + (1-\lambda)\pi' \right) \geq \lambda f(\pi) + (1-\lambda)f(\pi').\]
Moreover, if the Bayes rule $d_{\pi'}$ exists:
\[ f(\pi') \geq f(\pi) + \left( \int_{[-m,m]}R(d_{\pi'},\theta)d\pi' - \int_{[-m,m]}R(d_{\pi'},\theta)d\pi \right)  \]
for any $\pi \in \Delta([-m,m])$.
\end{lemma}

\begin{proof}
Let $\pi,\pi'$ be two arbitrary elements of $\Delta([-m,m])$. For any $\lambda \in [0,1]$, $\lambda \pi + (1-\lambda) \pi'$ is also an element of $\Delta([-m,m])$. By definition
\begin{eqnarray*}
f\left( \lambda \pi + (1-\lambda)\pi' \right) &=& \inf_{d \in \mathcal{D}} \int_{[-m,m]} R(d,\theta) d \left( \lambda \pi + (1-\lambda)\pi' \right) \\
&=& \inf_{d \in \mathcal{D}} \left( \lambda\int_{[-m,m]} R(d,\theta) d\pi + (1-\lambda)\int_{[-m,m]} R(d,\theta) d\pi'  \right) \\
&\geq& \lambda \inf_{d \in \mathcal{D}} \int_{[-m,m]} R(d,\theta) d\pi + (1-\lambda) \inf_{d \in \mathcal{D}} \int_{[-m,m]} R(d,\theta) d\pi' \\
&=& \lambda f(\pi) + (1-\lambda) f(\pi').
\end{eqnarray*}
To show the last part of the lemma, note that
\begin{eqnarray*}
f(\pi') &=& \int_{[-m,m]} R(d_{\pi'},\theta) d \pi' \\
&=& \left( \int_{[-m,m]} R(d_{\pi'},\theta) d \pi' - \int_{[-m,m]} R(d_{\pi'},\theta) d \pi \right)+ \int_{[-m,m]} R(d_{\pi'},\theta) d \pi \\
& \geq& \left( \int_{[-m,m]} R(d_{\pi'},\theta) d \pi' - \int_{[-m,m]} R(d_{\pi'},\theta) d \pi \right)+ \inf_{d \in \mathcal{D}} \int_{[-m,m]} R(d,\theta) d \pi \\
&=& \left( \int_{[-m,m]} R(d_{\pi'},\theta) d \pi' - \int_{[-m,m]} R(d_{\pi'},\theta) d \pi \right)+ f(\pi). 
\end{eqnarray*}
\end{proof} 

Lemma \ref{lemma:OptimizedBayesRisk is Lips} and Lemma \ref{lemma:OptimizedBayesRule is concave} show that the problem of finding a least-favorable distribution is equivalent to the problem of maximizing a concave, Lipschitz function over the space of all Borel probability measures over $[-m,m]$. This is an infinite-dimensional concave maximization problem over $\Delta([-m,m])$. 

\begin{lemma} [Sequential continuity of $d_{\pi}(y)$ at $\pi$ w.r.t. the 1-Wasserstein metric] \label{lemma:sequential continuity} Let $\pi \in \Delta([-m,m])$. Take any sequence $\pi_n$ such that 
\[ W(\pi_n,\pi) \overset{p}{\rightarrow} 0.    \]
Then, 
\[ |d_{\pi_n}(y)-d_{\pi}(y) | \overset{p}{\rightarrow} 0  \]
for almost every $y \in \mathbb{R}$. 
    
\end{lemma}

\begin{proof}
The Bayes rule $d_{\pi}(\cdot)$ for the Bounded Normal Mean problem is the posterior mean function:
\[d_{\pi}(y) := \frac{a(y,\pi)}{b(y,\pi)}, \]
where 
\[  a(y,\pi) := \int_{[-m,m]} \theta\phi(y;\theta) d\pi, \quad b(y,\pi):= \int_{[-m,m]} \phi(y;\theta) d\pi,\]
and $\phi(y;\theta)$ is the p.d.f. of a $\mathcal{N}(\theta,1)$ evaluated at $y$. Algebra shows that for any $y \in \mathbb{R}$ and any priors $\pi_n,\pi \in \Delta([-m,m])$ 
\begin{eqnarray*}
d_{\pi_n}(y)-d_{\pi}(y) &=&  \frac{a(y,\pi_n)}{b(y,\pi_n)} - \frac{a(y,\pi)}{b(y,\pi)} \\
&=& \frac{a(y,\pi_n) b(y,\pi)-a(y,\pi)b(y,\pi_n)}{b(y,\pi_n)b(y,\pi)} \\
&=& \frac{a(y,\pi_n) b(y,\pi)-a(y,\pi) b(y,\pi)+a(y,\pi) b(y,\pi)- a(y,\pi)b(y,\pi_n)}{b(y,\pi_n)b(y,\pi)} \\
&=&   \frac{\left[a(y,\pi_n)-a(y,\pi) \right] b(y,\pi)+ \left[ b(y,\pi) - b(y,\pi_n) \right]a(y,\pi)}{b(y,\pi_n)b(y,\pi)} \\
&=& \frac{1}{b(y,\pi_n)} [a(y,\pi_n)-a(y,\pi) ]  +  d_{\pi}(y) \frac{1}{b(y,\pi_n)} [ b(y,\pi) - b(y,\pi_n)]. 
\end{eqnarray*}
Note that the functions 
\[ \varphi_1(\theta; y) := \theta \phi(y;\theta) \quad \textrm{ and } \quad \varphi_2(\theta; y) := \phi(y;\theta),  \]
are both Lipschitz continuous over $\theta \in [-m,m]$ with constants 
\[ C_1(y) := \frac{1+m(|y|+m)}{\sqrt{2 \pi}}, \quad C_2(y):=\frac{|y|+m}{\sqrt{2 \pi}},\]
respectively. This means that, using the definition of the 1-Wasserstein metric in \eqref{eqn:KR-duality},
\begin{eqnarray*}
|d_{\pi_n}(y)-d_{\pi}(y) | &\leq&  \frac{1}{b(y,\pi_n)} \left( C_1(y)  +  m C_2(y) \right) W(\pi_n,\pi) \\
&=& \frac{1}{b(y,\pi_n)-b(y,\pi) + b(y,\pi)} \left( C_1(y)  +  m C_2(y) \right) W(\pi_n,\pi) \\
&\leq& \frac{1}{-C_2(y)W(\pi_n,\pi) + b(y,\pi)} \left( C_1(y)  +  m C_2(y) \right) W(\pi_n,\pi).
\end{eqnarray*}
Since $b(y,\pi)>0$ for any $y \in \mathbb{R}$, this means that for any sequence $\pi_n$ such that $W(\pi_n,\pi) \overset{p}{\rightarrow} 0$, we have that
\[ |d_{\pi_n}(y)-d_{\pi}(y) | \overset{p}{\rightarrow} 0  \]
for almost every $y \in \mathbb{R}$. 
\end{proof}

\section{Additional Simulation Results} \label{sec: appendix_d}

\subsection{Theoretical and empirical Lipschitz constants}

Our theoretical results establish that for any decision rule $d$ taking values
in $[-m,m]$, the risk function is Lipschitz on $[-m,m]$ with constant $L(m)=3m^2+4m$, that is,
\[|R(d,\theta)-R(d,\theta')| \le (3m^2+4m)|\theta-\theta'| \quad \forall  \theta,\theta' \in[ -m,m].\]
This bound is uniform over all admissible decision rules and is used to control the difference between the worst-case risk over the full interval and the maximum risk over a discretized grid of $[-m, m]$.

To assess the conservativeness of this bound, we also report an empirical estimate of this quantity evaluated as 
\[\widehat{L}_{\mathrm{grid}}(m) = \max_{i<j} \frac{|\widehat R(d_{\pi^\epsilon},\theta_i)-\widehat R(d_{\pi^\epsilon},\theta_j)|
}{|\theta_i-\theta_j|},\]
where $\widehat{R}(d_{\pi^\epsilon},\theta_i)$ is the Monte Carlo estimate of the risk at the grid point $\theta_i$. Figure \ref{fig: lipschitz_constant} illustrates this comparison.

\begin{figure}[h!] 
\begin{center}	
\includegraphics[width=5in]{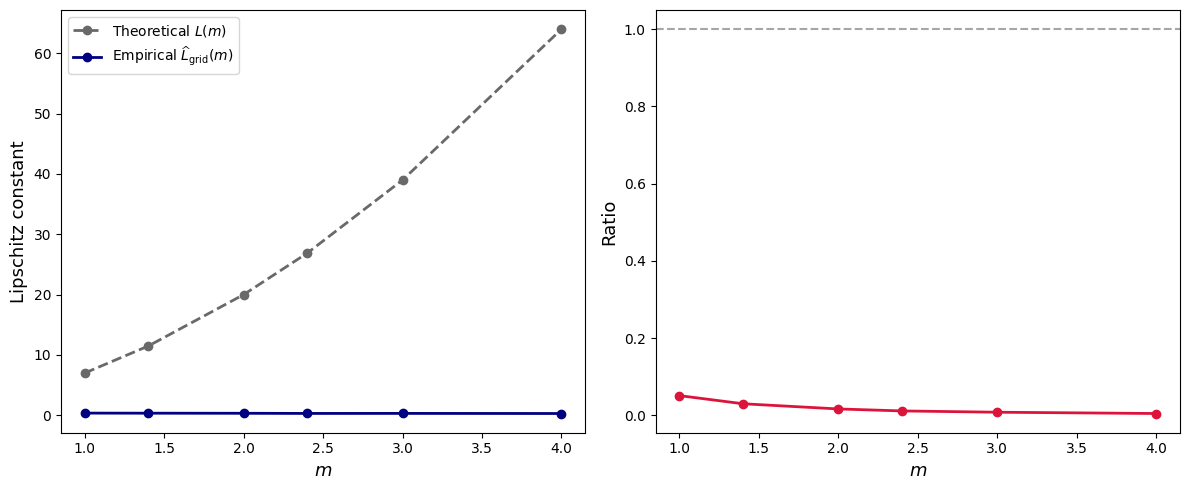}
\caption{Comparison of theoretical and empirical Lipschitz constants.}
\label{fig: lipschitz_constant}
\end{center}	
\end{figure}

\subsection{Number of iterations under theoretical vs. realized supergradients} 

\begin{figure}[h!] 
\begin{center}	
\includegraphics[width=3.75in]{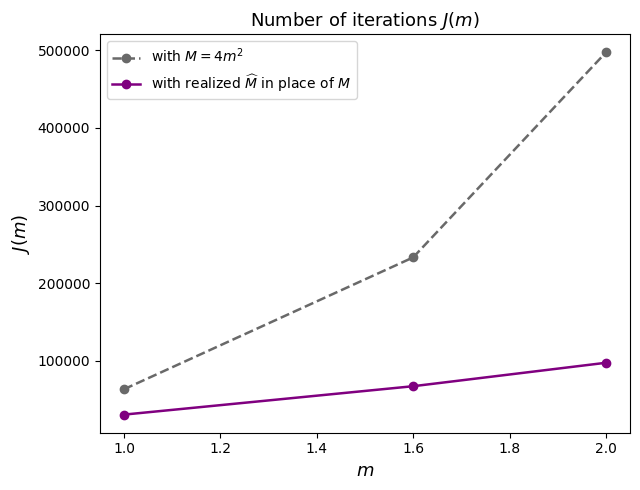}
\caption{Comparison of the number of iterations using the worst-case bound $M$ and its empirical counterpart.}
\label{fig: J_of_M}
\end{center}	
\end{figure}

\subsection{A certified upper bound on the worst-case risk} \label{subsec:grid_vs_interval_risk}

The worst-case risk evaluated on the finite grid is a priori lower than the true maximal value, since the supremum may be attained between grid points. To certify the gap, we use Lipschitz continuity of the risk function.

For any $\theta \in [\theta_i, \theta_{i+1}]$, applying Lemma \ref{lemma:RiskFunction is Lips} at each endpoint and averaging gives
\[R(d_{\pi^\epsilon}, \theta) \le \frac{1}{2} \Big[ R(d_{\pi^\epsilon},\theta_i) + (3m^2+4m) |\theta - \theta_i| \Big] + \frac{1}{2} \Big[ R(d_{\pi^\epsilon},\theta_{i+1}) + (3m^2+4m) |\theta - \theta_{i+1}| \Big]. \]

Since $|\theta - \theta_i| + | \theta - \theta_{i+1}| = \mathbf{\Delta}(m,I)$, taking the supremum over $\theta \in [-m,m]$
\[\sup_{\theta \in [-m,m]} R(d_{\pi^\epsilon},\theta) \le \min \left\{ \max_{i \in \{1, \dots, I-1\}} \frac{R(d_{\pi^\epsilon},\theta_i) + R(d_{\pi^\epsilon},\theta_{i+1})}{2}  + \frac{(3m^2+4m) \mathbf{\Delta}(m,I)}{2}, 4m^2 \right\}, \]
where $\mathbf{\Delta}(m,I) = 2m/(I-1)$ is the grid spacing.

Figure \ref{fig: risk_all_m_4_certified} applies an analogous Lipschitz-certification to the truncated maximum likelihood estimator $d_{\textrm{trunc}} (y) = \max\{\min\{y,m\}, -m\}$ (green line) and the \emph{clipped} linear estimator (purple line) $d_{\textrm{linear, clipped}} = \max\{\min\{c \cdot y,m\}, -m\}$, where $c=m^2/(1+m^2)$. All risks are evaluated with 1,000 Monte Carlo draws over the 200,000-point grid.

\begin{figure}[h!] 
\begin{center}	
\includegraphics[width=3.75in]{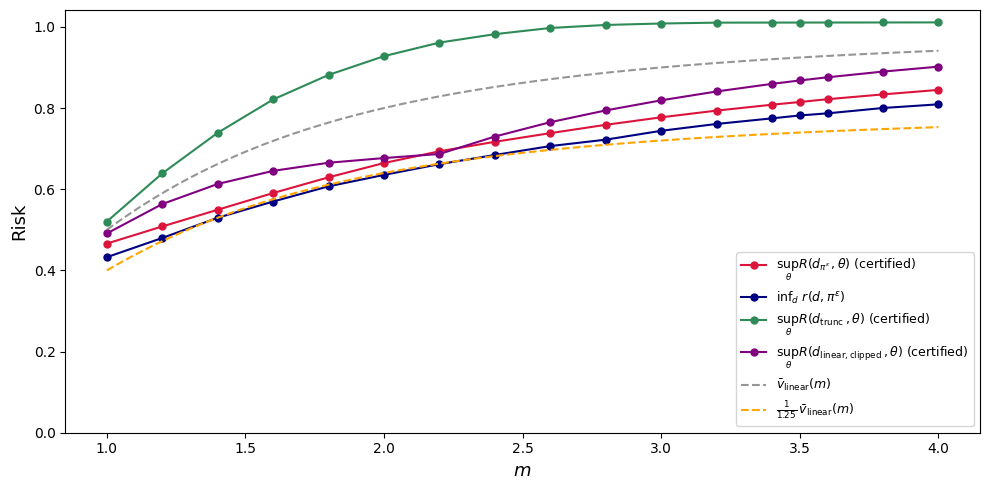}
\caption{Worst-case risk of $d_{\pi^\epsilon}$ (certified bound) vs. alternative estimators and analytical bounds.}
\label{fig: risk_all_m_4_certified}
\end{center}	
\end{figure}

\subsection{Comparison with \citet{casella_strawderman_1981}'s closed-form for $m \le 1$} \label{subsec:casella_strawderman_risk}

For $m \le 1$ the minimax risk is available in closed form (Theorems 3.1 and 4.2 in \citet{casella_strawderman_1981}). Figure~\ref{fig: risk_m_1} shows that the interval $[\inf_{d} r(d,\pi^\epsilon), \sup_{\theta} R(d_{\pi^\epsilon},\theta)]$ produced by our algorithm contains the true minimax value $\bar{v}^*(m)$ at every point, validating the procedure where the exact answer is known.

\begin{figure}[H] 
\begin{center}	
\includegraphics[width=3.75in]{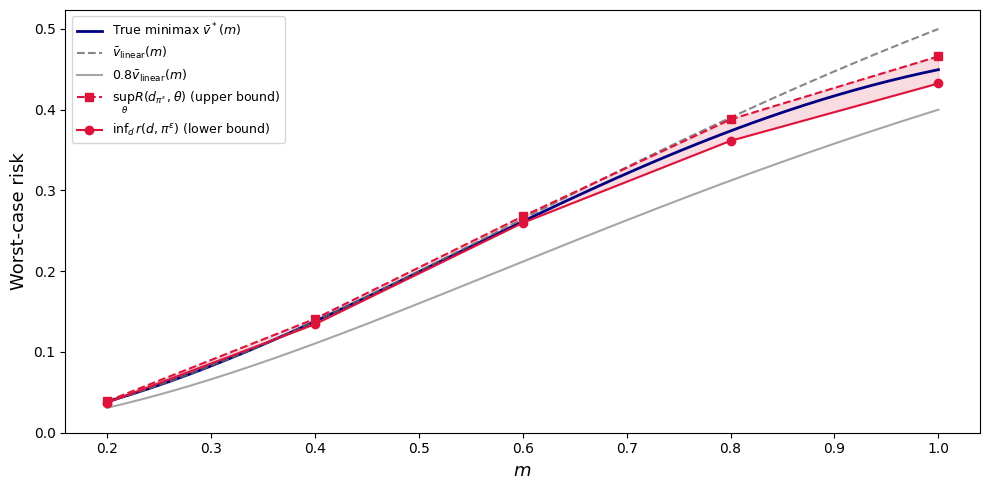}
\caption{$d_{\pi^\epsilon}$ vs. linear vs. true minimax.}
\label{fig: risk_m_1}
\end{center}	
\end{figure}

\subsection{Stability of the risk evaluation under Monte Carlo randomness}

One concern with the estimator $d_{\pi^\epsilon}$ is that---due to randomness introduced by the stochastic mirror ascent procedure---different runs of the algorithm could generate very different estimators. One way to assuage this concern is to report the behavior of the risk function across different seeds. Figure \ref{fig: risk_of_theta_2} displays the risk $R(d_{\pi^\epsilon}, \theta)$ for different seeds evaluated for $\theta \in [-2,2]$. The shaded region in the graph represents the area between the 2.5\% and the 97.5\% quantiles across different seeds. The solid curve shows the across-seed mean. Each profile is estimated by Monte Carlo over 100 independent seeds with 10,000 draws per seed.

\begin{figure}[h!] 
\begin{center}	
\includegraphics[width=3.75in]{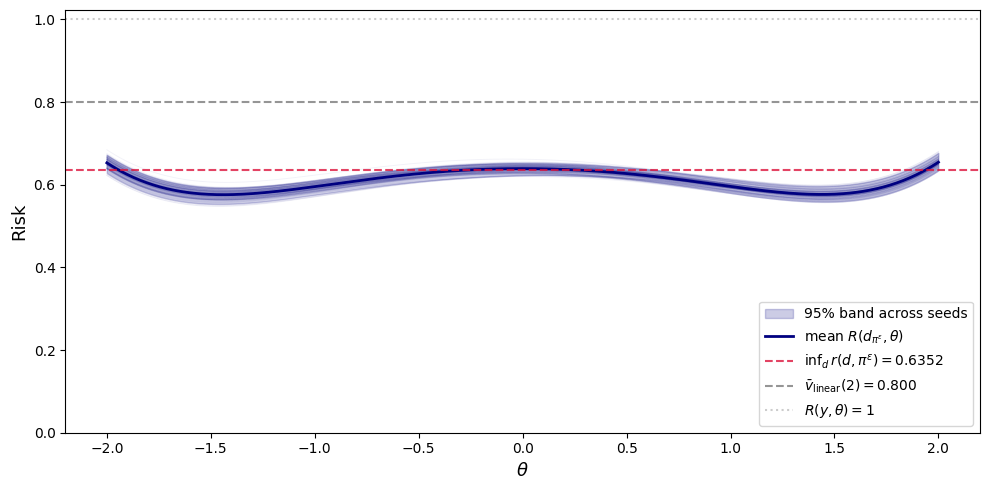}
\caption{Risk of $d_{\pi^\epsilon}$ for $m = 2$.}
\label{fig: risk_of_theta_2}
\end{center}	
\end{figure}

\end{appendix}


\begin{thebibliography}{61}
\newcommand{\enquote}[1]{``#1''}
\expandafter\ifx\csname natexlab\endcsname\relax\def\natexlab#1{#1}\fi

\bibitem[\protect\citeauthoryear{Aradillas~Fern{\'a}ndez, Blanchet, Montiel~Olea, Qiu, Stoye, and Tan}{Aradillas~Fern{\'a}ndez et~al.}{2025{\natexlab{a}}}]{fernandez2025approximate}
\textsc{Aradillas~Fern{\'a}ndez, A., J.~Blanchet, J.~L. Montiel~Olea, C.~Qiu, J.~Stoye, and L.~Tan} (2025{\natexlab{a}}): \enquote{Approximate least-favorable distributions and nearly optimal tests via stochastic mirror descent,} Preprint, arXiv:2511.16925.

\bibitem[\protect\citeauthoryear{Aradillas~Fern{\'a}ndez, Blanchet, Montiel~Olea, Qiu, Stoye, and Tan}{Aradillas~Fern{\'a}ndez et~al.}{2025{\natexlab{b}}}]{montielolea_aradillasfernandez_blanchet_qiu_stoye_tan_2024}
---\hspace{-.1pt}---\hspace{-.1pt}--- (2025{\natexlab{b}}): \enquote{Epsilon-Minimax Solutions of Statistical Decision Problems,} Preprint, arXiv:2509.08107.

\bibitem[\protect\citeauthoryear{Armstrong, Kline, and Sun}{Armstrong et~al.}{2025}]{armstrong_kline_sun_2025}
\textsc{Armstrong, T.~B., P.~Kline, and L.~Sun} (2025): \enquote{Adapting to Misspecification,} \emph{Econometrica}, 93, 1981--2005.

\bibitem[\protect\citeauthoryear{Barsov and Ulyanov}{Barsov and Ulyanov}{1987}]{barsov1987proximity}
\textsc{Barsov, S.~S. and V.~V. Ulyanov} (1987): \enquote{Estimates of the Proximity of Gaussian Measures,} \emph{Soviet Mathematics Doklady}, 34, 462--466.

\bibitem[\protect\citeauthoryear{Ben-Tal, Margalit, and Nemirovski}{Ben-Tal et~al.}{2001}]{ben2001ordered}
\textsc{Ben-Tal, A., T.~Margalit, and A.~Nemirovski} (2001): \enquote{The ordered subsets mirror descent optimization method with applications to tomography,} \emph{SIAM Journal on Optimization}, 12, 79--108.

\bibitem[\protect\citeauthoryear{Bickel}{Bickel}{1981}]{bickel1981minimax}
\textsc{Bickel, P.~J.} (1981): \enquote{Minimax estimation of the mean of a normal distribution when the parameter space is restricted,} \emph{The Annals of Statistics}, 9, 1301--1309.

\bibitem[\protect\citeauthoryear{Blackwell and Girshick}{Blackwell and Girshick}{1954}]{blackwell1954}
\textsc{Blackwell, D. and M.~A. Girshick} (1954): \emph{Theory of games and statistical decisions}, New York: Wiley.

\bibitem[\protect\citeauthoryear{Blanchet, Kang, Montiel~Olea, Nguyen, and Zhang}{Blanchet et~al.}{2023}]{JMLR:v24:21-0377}
\textsc{Blanchet, J., Y.~Kang, J.~L. Montiel~Olea, V.~A. Nguyen, and X.~Zhang} (2023): \enquote{Dropout Training is Distributionally Robust Optimal,} \emph{Journal of Machine Learning Research}, 24, 1--60.

\bibitem[\protect\citeauthoryear{Boucheron, Lugosi, and Massart}{Boucheron et~al.}{2013}]{boucheron2013}
\textsc{Boucheron, S., G.~Lugosi, and P.~Massart} (2013): \emph{Concentration Inequalities: A Nonasymptotic Theory of Independence}, Oxford University Press.

\bibitem[\protect\citeauthoryear{Bubeck}{Bubeck}{2015}]{bubeck_2015}
\textsc{Bubeck, S.} (2015): \enquote{Convex Optimization: Algorithms and Complexity,} \emph{Foundations and Trends in Machine Learning}, 8, 231--357.

\bibitem[\protect\citeauthoryear{Casella and Strawderman}{Casella and Strawderman}{1981}]{casella_strawderman_1981}
\textsc{Casella, G. and W.~E. Strawderman} (1981): \enquote{Estimating a Bounded Normal Mean,} \emph{The Annals of Statistics}, 9, 870--878.

\bibitem[\protect\citeauthoryear{Cattaneo, Klusowski, and Shigida}{Cattaneo et~al.}{2024}]{cattaneo2024onimplicitbias}
\textsc{Cattaneo, M.~D., J.~M. Klusowski, and B.~Shigida} (2024): \enquote{On the Implicit Bias of Adam,} in \emph{Forty-first International Conference on Machine Learning}.

\bibitem[\protect\citeauthoryear{Cattaneo and Shigida}{Cattaneo and Shigida}{2025{\natexlab{a}}}]{cattaneo2025howmemory}
\textsc{Cattaneo, M.~D. and B.~Shigida} (2025{\natexlab{a}}): \enquote{How Memory in Optimization Algorithms Implicitly Modifies the Loss,} in \emph{The Thirty-ninth Annual Conference on Neural Information Processing Systems}.

\bibitem[\protect\citeauthoryear{Cattaneo and Shigida}{Cattaneo and Shigida}{2025{\natexlab{b}}}]{cattaneo2025modifiedlossmomentum}
---\hspace{-.1pt}---\hspace{-.1pt}--- (2025{\natexlab{b}}): \enquote{Modified Loss of Momentum Gradient Descent: Fine-Grained Analysis,} Preprint, arXiv:2509.08483.

\bibitem[\protect\citeauthoryear{Chamberlain}{Chamberlain}{2000}]{chamberlain2000econometric}
\textsc{Chamberlain, G.} (2000): \enquote{Econometric applications of maxmin expected utility,} \emph{Journal of Applied Econometrics}, 15, 625--644.

\bibitem[\protect\citeauthoryear{Chandrasekaran and Jordan}{Chandrasekaran and Jordan}{2013}]{chandrasekaran2013computational}
\textsc{Chandrasekaran, V. and M.~I. Jordan} (2013): \enquote{Computational and statistical tradeoffs via convex relaxation,} \emph{Proceedings of the National Academy of Sciences}, 110, E1181--E1190.

\bibitem[\protect\citeauthoryear{Chen, Pesavento, and Vonn{\'a}k}{Chen et~al.}{2026}]{chen_pesavento_vonnak_2026}
\textsc{Chen, C., E.~Pesavento, and B.~Vonn{\'a}k} (2026): \enquote{Estimator Averaging of Local Projection and {VAR} Impulse Responses,} Preprint, arXiv:2605.05456.

\bibitem[\protect\citeauthoryear{Delatte}{Delatte}{2026}]{delatte26}
\textsc{Delatte, P.} (2026): \enquote{Some existence results for maximin priors in statistical minimax theorems,} Working paper, University of Southern California.

\bibitem[\protect\citeauthoryear{Devroye, Mehrabian, and Reddad}{Devroye et~al.}{2018}]{devroye2018total}
\textsc{Devroye, L., A.~Mehrabian, and T.~Reddad} (2018): \enquote{The total variation distance between high-dimensional Gaussians with the same mean,} Preprint, arXiv:1810.08693.

\bibitem[\protect\citeauthoryear{Donoho, Liu, and MacGibbon}{Donoho et~al.}{1990}]{donoho_liu_macgibbon_1990}
\textsc{Donoho, D.~L., R.~C. Liu, and B.~MacGibbon} (1990): \enquote{Minimax Risk Over Hyperrectangles, and Implications,} \emph{The Annals of Statistics}, 18, 1416--1437.

\bibitem[\protect\citeauthoryear{Duchi, Hazan, and Singer}{Duchi et~al.}{2011}]{duchi2011}
\textsc{Duchi, J., E.~Hazan, and Y.~Singer} (2011): \enquote{Adaptive Subgradient Methods for Online Learning and Stochastic Optimization,} \emph{Journal of Machine Learning Research}, 12, 2121--2159.

\bibitem[\protect\citeauthoryear{Elliott, M{\"u}ller, and Watson}{Elliott et~al.}{2015}]{elliott2015nearly}
\textsc{Elliott, G., U.~K. M{\"u}ller, and M.~W. Watson} (2015): \enquote{Nearly optimal tests when a nuisance parameter is present under the null hypothesis,} \emph{Econometrica}, 83, 771--811.

\bibitem[\protect\citeauthoryear{Feldman and Brown}{Feldman and Brown}{1989}]{feldman1989minimax}
\textsc{Feldman, I. and L.~D. Brown} (1989): \enquote{The minimax risk for estimating a bounded normal mean,} Tech. rep., Cornell University, Statistics Center.

\bibitem[\protect\citeauthoryear{Ferguson}{Ferguson}{1967}]{ferguson_1967}
\textsc{Ferguson, T.~S.} (1967): \emph{Mathematical Statistics: A Decision Theoretic Approach}, New York: Academic Press.

\bibitem[\protect\citeauthoryear{Francis, Owyang, Roush, and DiCecio}{Francis et~al.}{2014}]{francis_owyang_roush_dicecio_2014}
\textsc{Francis, N., M.~T. Owyang, J.~E. Roush, and R.~DiCecio} (2014): \enquote{A Flexible Finite-Horizon Alternative to Long-Run Restrictions with an Application to Technology Shocks,} \emph{Review of Economics and Statistics}, 96, 638--647.

\bibitem[\protect\citeauthoryear{Gertler and Karadi}{Gertler and Karadi}{2015}]{gertler_karadi_2015}
\textsc{Gertler, M. and P.~Karadi} (2015): \enquote{Monetary Policy Surprises, Credit Costs, and Economic Activity,} \emph{American Economic Journal: Macroeconomics}, 7, 44--76.

\bibitem[\protect\citeauthoryear{Ghosh}{Ghosh}{1964}]{ghosh1964uniform}
\textsc{Ghosh, M.~N.} (1964): \enquote{Uniform approximation of minimax point estimates,} \emph{The Annals of Mathematical Statistics}, 35, 1031--1047.

\bibitem[\protect\citeauthoryear{Gourdin, Jaumard, and MacGibbon}{Gourdin et~al.}{1994}]{gourdin1994global}
\textsc{Gourdin, {\'E}., B.~Jaumard, and B.~MacGibbon} (1994): \enquote{Global optimization decomposition methods for bounded parameter minimax risk evaluation,} \emph{SIAM Journal on Scientific Computing}, 15, 16--35.

\bibitem[\protect\citeauthoryear{Guggenberger and Huang}{Guggenberger and Huang}{2025}]{guggenberger2025numerical}
\textsc{Guggenberger, P. and J.~Huang} (2025): \enquote{On the numerical approximation of minimax regret rules via fictitious play,} Preprint, arXiv:2503.10932.

\bibitem[\protect\citeauthoryear{Hausman}{Hausman}{1978}]{hausman1978specification}
\textsc{Hausman, J.~A.} (1978): \enquote{Specification tests in econometrics,} \emph{Econometrica}, 46, 1251--1271.

\bibitem[\protect\citeauthoryear{Ibragimov and Has'minskii}{Ibragimov and Has'minskii}{1985}]{ibragimov1985nonparametric}
\textsc{Ibragimov, I.~A. and R.~Z. Has'minskii} (1985): \enquote{On nonparametric estimation of the value of a linear functional in Gaussian white noise,} \emph{Theory of Probability \& Its Applications}, 29, 18--32.

\bibitem[\protect\citeauthoryear{Johnstone}{Johnstone}{2019}]{johnstone_2019}
\textsc{Johnstone, I.~M.} (2019): \enquote{Gaussian Estimation: Sequence and Wavelet Models,} Draft version, September 16, 2019.

\bibitem[\protect\citeauthoryear{Johnstone and MacGibbon}{Johnstone and MacGibbon}{1992}]{johnstone1992minimax}
\textsc{Johnstone, I.~M. and K.~B. MacGibbon} (1992): \enquote{Minimax estimation of a constrained Poisson vector,} \emph{The Annals of Statistics}, 20, 807--831.

\bibitem[\protect\citeauthoryear{Jord{\`a}}{Jord{\`a}}{2005}]{jorda_2005}
\textsc{Jord{\`a}, {\`O}.} (2005): \enquote{Estimation and Inference of Impulse Responses by Local Projections,} \emph{American Economic Review}, 95, 161--182.

\bibitem[\protect\citeauthoryear{Jord{\`a} and Taylor}{Jord{\`a} and Taylor}{2025}]{jorda2025local}
\textsc{Jord{\`a}, {\`O}. and A.~M. Taylor} (2025): \enquote{Local projections,} \emph{Journal of Economic Literature}, 63, 59--110.

\bibitem[\protect\citeauthoryear{Juditsky, Rigollet, and Tsybakov}{Juditsky et~al.}{2008}]{juditsky2008}
\textsc{Juditsky, A., P.~Rigollet, and A.~B. Tsybakov} (2008): \enquote{Learning by mirror averaging,} \emph{The Annals of Statistics}, 36, 2183--2206.

\bibitem[\protect\citeauthoryear{Kempthorne}{Kempthorne}{1987}]{kempthorne1987numerical}
\textsc{Kempthorne, P.~J.} (1987): \enquote{Numerical specification of discrete least favorable prior distributions,} \emph{SIAM Journal on Scientific and Statistical Computing}, 8, 171--184.

\bibitem[\protect\citeauthoryear{Kent, Li, Blanchet, and Glynn}{Kent et~al.}{2021}]{kent2021modified}
\textsc{Kent, C., J.~Li, J.~Blanchet, and P.~W. Glynn} (2021): \enquote{Modified Frank Wolfe in probability space,} in \emph{Advances in Neural Information Processing Systems (NeurIPS)}, vol.~34, 14448--14462.

\bibitem[\protect\citeauthoryear{Krichene, Balandat, Tomlin, and Bayen}{Krichene et~al.}{2015}]{krichene2015hedge}
\textsc{Krichene, W., M.~Balandat, C.~Tomlin, and A.~Bayen} (2015): \enquote{The hedge algorithm on a continuum,} in \emph{International Conference on Machine Learning}, PMLR, vol.~37, 824--832.

\bibitem[\protect\citeauthoryear{Lei and Jordan}{Lei and Jordan}{2020}]{leijordan2020}
\textsc{Lei, L. and M.~I. Jordan} (2020): \enquote{On the Adaptivity of Stochastic Gradient-Based Optimization,} \emph{SIAM Journal on Optimization}, 30, 1473--1500.

\bibitem[\protect\citeauthoryear{Levit}{Levit}{1981}]{levit1981asymptotic}
\textsc{Levit, B.~Y.} (1981): \enquote{On asymptotic minimax estimates of the second order,} \emph{Theory of Probability \& Its Applications}, 25, 552--568.

\bibitem[\protect\citeauthoryear{Li, Plagborg-M{\o}ller, and Wolf}{Li et~al.}{2024}]{li2024local}
\textsc{Li, D., M.~Plagborg-M{\o}ller, and C.~K. Wolf} (2024): \enquote{Local projections vs. VARs: Lessons from thousands of DGPs,} \emph{Journal of Econometrics}, 244, 105722.

\bibitem[\protect\citeauthoryear{Lin, Bickel, and Ding}{Lin et~al.}{2026}]{lin2026introducing}
\textsc{Lin, Z., P.~J. Bickel, and P.~Ding} (2026): \enquote{Introducing the b-value: combining unbiased and biased estimators from a sensitivity analysis perspective,} Preprint, arXiv:2602.16310.

\bibitem[\protect\citeauthoryear{Manski}{Manski}{2021}]{manski2021econometrics}
\textsc{Manski, C.~F.} (2021): \enquote{Econometrics for decision making: Building foundations sketched by Haavelmo and Wald,} \emph{Econometrica}, 89, 2827--2853.

\bibitem[\protect\citeauthoryear{Montiel~Olea and Plagborg-M{\o}ller}{Montiel~Olea and Plagborg-M{\o}ller}{2021}]{montielolea_plagborgmoller_2021}
\textsc{Montiel~Olea, J.~L. and M.~Plagborg-M{\o}ller} (2021): \enquote{Local Projection Inference is Simpler and More Robust Than You Think,} \emph{Econometrica}, 89, 1789--1823.

\bibitem[\protect\citeauthoryear{Montiel~Olea, Plagborg-M{\o}ller, Qian, and Wolf}{Montiel~Olea et~al.}{2025}]{montielolea_qian_wolf_plagborgmoller_2025}
\textsc{Montiel~Olea, J.~L., M.~Plagborg-M{\o}ller, E.~Qian, and C.~K. Wolf} (2025): \enquote{Local Projections or VARs? A Primer for Macroeconomists,} \emph{NBER Macroeconomics Annual}, 40, 111--152.

\bibitem[\protect\citeauthoryear{Montiel~Olea, Plagborg-M{\o}ller, Qian, and Wolf}{Montiel~Olea et~al.}{2026}]{montielolea_plagborgmoller_qian_wolf_forthcoming}
---\hspace{-.1pt}---\hspace{-.1pt}--- (2026): \enquote{Double Robustness of Local Projections and Some Unpleasant {VAR}ithmetic,} \emph{Econometrica}, forthcoming.

\bibitem[\protect\citeauthoryear{Negrea, Bilodeau, Campolongo, Orabona, and Roy}{Negrea et~al.}{2021}]{negrea2021minimax}
\textsc{Negrea, J., B.~Bilodeau, N.~Campolongo, F.~Orabona, and D.~Roy} (2021): \enquote{Minimax optimal quantile and semi-adversarial regret via root-logarithmic regularizers,} in \emph{Advances in Neural Information Processing Systems (NeurIPS)}, vol.~34, 26237--26249.

\bibitem[\protect\citeauthoryear{Nemirovski, Juditsky, Lan, and Shapiro}{Nemirovski et~al.}{2009}]{nemirovski2009robust}
\textsc{Nemirovski, A., A.~Juditsky, G.~Lan, and A.~Shapiro} (2009): \enquote{Robust stochastic approximation approach to stochastic programming,} \emph{SIAM Journal on Optimization}, 19, 1574--1609.

\bibitem[\protect\citeauthoryear{Nemirovski and Yudin}{Nemirovski and Yudin}{1983}]{nemirovski_yudin_1983}
\textsc{Nemirovski, A. and D.~Yudin} (1983): \emph{Problem Complexity and Method Efficiency in Optimization}, New York: Wiley.

\bibitem[\protect\citeauthoryear{Nemtyrev and Boldea}{Nemtyrev and Boldea}{2026}]{nemtyrev2026targeted}
\textsc{Nemtyrev, A. and O.~Boldea} (2026): \enquote{Targeted Local Projections,} Preprint, arXiv:2603.00248.

\bibitem[\protect\citeauthoryear{Plagborg-M{\o}ller and Wolf}{Plagborg-M{\o}ller and Wolf}{2021}]{plagborgmoller_wolf_2021}
\textsc{Plagborg-M{\o}ller, M. and C.~K. Wolf} (2021): \enquote{Local Projections and VARs Estimate the Same Impulse Responses,} \emph{Econometrica}, 89, 955--980.

\bibitem[\protect\citeauthoryear{Ramey}{Ramey}{2011}]{ramey_2011}
\textsc{Ramey, V.~A.} (2011): \enquote{Identifying Government Spending Shocks: It's all in the Timing,} \emph{Quarterly Journal of Economics}, 126, 1--50.

\bibitem[\protect\citeauthoryear{Ramey}{Ramey}{2016}]{ramey_2016}
---\hspace{-.1pt}---\hspace{-.1pt}--- (2016): \enquote{Macroeconomic Shocks and Their Propagation,} in \emph{Handbook of Macroeconomics}, ed. by J.~B. Taylor and H.~Uhlig, Elsevier, vol.~2, chap.~2, 71--162.

\bibitem[\protect\citeauthoryear{Romer and Romer}{Romer and Romer}{2010}]{romer_romer_2010}
\textsc{Romer, C.~D. and D.~H. Romer} (2010): \enquote{The Macroeconomic Effects of Tax Changes: Estimates Based on a New Measure of Fiscal Shocks,} \emph{American Economic Review}, 100, 763--801.

\bibitem[\protect\citeauthoryear{Sims}{Sims}{1980}]{sims_1980}
\textsc{Sims, C.~A.} (1980): \enquote{Macroeconomics and Reality,} \emph{Econometrica}, 48, 1--48.

\bibitem[\protect\citeauthoryear{Stock and Watson}{Stock and Watson}{2018}]{stock2018identification}
\textsc{Stock, J.~H. and M.~W. Watson} (2018): \enquote{Identification and estimation of dynamic causal effects in macroeconomics using external instruments,} \emph{The Economic Journal}, 128, 917--948.

\bibitem[\protect\citeauthoryear{Villani}{Villani}{2021}]{villani2021topics}
\textsc{Villani, C.} (2021): \emph{Topics in optimal transportation}, Graduate Studies in Mathematics, American Mathematical Society.

\bibitem[\protect\citeauthoryear{Williamson and Shmoys}{Williamson and Shmoys}{2011}]{shmoys2011design}
\textsc{Williamson, D.~P. and D.~B. Shmoys} (2011): \emph{The design of approximation algorithms}, Cambridge University Press.

\bibitem[\protect\citeauthoryear{Wu, Salazar, Green, and Blei}{Wu et~al.}{2026}]{wuetal2026}
\textsc{Wu, B., S.~Salazar, D.~P. Green, and D.~M. Blei} (2026): \enquote{The Illusion of Learning from Observational Data: An Empirical Bayes Perspective,} Preprint, arXiv:2604.08853.

\bibitem[\protect\citeauthoryear{Xu}{Xu}{2026}]{xu2026local}
\textsc{Xu, K.-L.} (2026): \enquote{{Local Projection-Based Inference under General Conditions},} \emph{The Review of Economic Studies}, rdag034.

\end{thebibliography}
\end{document}